\documentclass[final,journal,a4paper,twocolumn]{IEEEtran}
\usepackage[singlelinecheck=false]{caption}
\usepackage{hyperref}   
\hypersetup{colorlinks=false,linkbordercolor=.75 .9 .95,citebordercolor=.75 .9 .95,urlbordercolor=.5 .7 .1}
\usepackage{cite}
\usepackage{amsmath,amssymb,amsfonts}
\usepackage{algorithmic}
\usepackage{graphicx}
\usepackage{textcomp}
\usepackage{xcolor}
\usepackage{enumitem}
\usepackage{amsthm}
\usepackage{mathtools}
\usepackage{subcaption}
\usepackage{multirow}
\RequirePackage{bm}
\RequirePackage{mathrsfs}
\usepackage{threeparttable}
\usepackage{diagbox} 
\graphicspath{{figures/}}
\usepackage{CJKutf8}
\usepackage{microtype}
\usepackage{afterpage}

\makeatletter
\renewcommand{\maketag@@@}[1]{\hbox{\m@th\normalsize\normalfont#1}}
\makeatother

\usepackage{float}
\usepackage{algorithm}
\floatstyle{plain} 
\newfloat{myairfloat}{htbp}{myair}
{	\begin{myairfloat}[#1]
		\centering
		\begin{minipage}{#2}
			\begin{algorithm}[H]}
			{\end{algorithm}
		\end{minipage}
	\vspace{-5mm}
	\end{myairfloat}
}

\allowdisplaybreaks[4]
\newlength{\Oldarrayrulewidth}

\def\BibTeX{{\rm B\kern-.05em{\sc i\kern-.025em b}\kern-.08em
    T\kern-.1667em\lower.7ex\hbox{E}\kern-.125emX}}

\theoremstyle{definition}
\newtheorem{theorem}{\bf Proposition}
\newtheorem{lemma}{\bf Lemma}

\newtheorem{remark}{Remark}

\begin{document}
\title{SA-WiSense: A Blind-Spot-Free Respiration Sensing Framework for Single-Antenna Wi-Fi Devices}
\author{Guangteng Liu, Xiayue Liu, Zhixiang Xu, Yufeng Yuan, Hui Zhao, Yuxuan Liu\\ Yufei Jiang,~\IEEEmembership{Member,~IEEE}
}
\maketitle
\begin{CJK}{UTF8}{gbsn}
        \begin{abstract}
Wi-Fi sensing offers a promising technique for contactless human respiration monitoring. A key challenge, however, is the blind spot problem caused by random phase offsets that corrupt the complementarity of respiratory signals. To address the challenge, we propose a single-antenna-Wi-Fi-sensing (SA-WiSense) framework to improve accuracy of human respiration monitoring, robust against random phase offsets. The proposed SA-WiSense framework is cost-efficient, as only a single antenna is used rather than multiple antennas as in the previous works. Therefore, the proposed framework is applicable to Internet of Thing (IoT), where most of sensors are equipped with a single antenna. On one hand, we propose a cross-subcarrier channel state information (CSI) ratio (CSCR) based blind spot mitigation approach for IoT, where the ratios of two values of CSI between subcarriers are leveraged to mitigate random phase offsets. We prove that the random phase offsets can be cancelled by the proposed CSCR approach, thereby restoring the inherent complementarity of signals for blind-spot-free sensing. On the other hand, we propose a genetic algorithm (GA) based subcarrier selection (GASS) approach by formulating an optimization problem in terms of the sensing-signal-to-noise ratio (SSNR) of CSCR between subcarriers. GA is utilized to solve the formulated optimization problem. We use commodity ESP32 microcontrollers to build an experiment test. The proposed works are validated to achieve an detection rate of 91.2\% for respiration monitoring at distances up to 8.0 meters, substantially more accurate than the state-of-the-art methods with a single antenna.
\end{abstract}


\vspace{-3mm}
\section{Introduction}
Respiration is a fundamental vital sign for monitoring disease progression \cite{heng2025exhaled}.
The respiratory rate, in particular, provides diagnostic information comparable to other key physiological parameters such as pulse rate and blood pressure \cite{Wang2024SlpRoF,Zhao2023}.
Abnormal breathing patterns, which may arise from various conditions including asthma, central sleep apnea, and emphysema, can significantly disrupt sleep quality. Such disruptions may lead to excessive daytime sleepiness and an elevated risk of cardiovascular disease \cite{bao2023wi}.
Consequently, the development of continuous and cost-effective respiration monitoring solutions for home-based environments holds substantial clinical significance.

\subsection{Related Work}
Technologies for residential respiration monitoring are categorized into contact-based and contactless approaches.
Contact-based methods, such as wearable sensors~\cite{Guo2024MagWear,Tajin2021}, offer high accuracy and provide reliable data for clinical diagnosis. However, the requirement for physical attachment to the user inherently causes discomfort and inconvenience, limiting the feasibility of such methods for long-term monitoring in a home environment~\cite{Song2025FinerSense,Guo2021Emergency}.
Consequently, contactless alternatives have gained significant attention. For instance, systems achieve remote monitoring using commodity cameras~\cite{Huang2024Camera} or depth cameras~\cite{Ottaviani2022Contactless}. However, the camera-based systems raise significant privacy concerns and rely on expensive, high-resolution equipment~\cite{Zhang2025mmTAA}. 
Other contactless solutions leverage radio frequency (RF) signals. 
While techniques based on Doppler radar have been widely investigated~\cite{Jang2024Remote}, the reliance on specialized and costly hardware renders the approaches impractical for scalable, low-cost applications.

Recently, Wi-Fi-based RF sensing has emerged as a promising non-contact technique that leverages the ubiquity of commodity wireless devices in residential environments.
Such existing infrastructure can be repurposed to detect subtle human motions, enabling applications such as activity recognition~\cite{zhang2022widar3} and respiration detection~\cite{yu2021wisleep}.
The non-invasive nature, privacy-preserving characteristics, and low deployment cost further establish Wi-Fi-based sensing as a compelling technique for long-term, in-home respiration monitoring.

While some systems utilizes received signal strength (RSS)~\cite{Yigitler2020RSS}, 
the fine-grained channel state information (CSI) provides superior sensing resolution, establishing a new state-of-the-art for respiration monitoring~\cite{zeng2018fullbreathe, wang2022research}.
Initial approaches, such as the one presented in~\cite{wang2016human}, achieves high-precision respiratory rate detection using only CSI amplitude. 
However, relying solely on amplitude introduces a fundamental challenge, as revealed by analysis based on Fresnel zone theory~\cite{wang2016human}: the existence of blind spots, where respiratory motion is undetectable.
A conceptual breakthrough in~\cite{zeng2018fullbreathe} established the complementarity between CSI amplitude and phase. The property enables the joint utilization of both components to robustly capture respiratory motion, thereby overcoming the blind spot problem inherent in using either component alone.
However, the raw phase extracted from commodity Wi-Fi devices is unreliable for direct use due to random offsets induced by hardware imperfections, such as clock asynchrony~\cite{li2021complexbeat, Yang2021Respiration}.
To address the challenge, solutions such as FarSense~\cite{zeng2019farsense} construct a CSI ratio that relies on the spatial diversity of multiple-input multiple-output (MIMO) systems.
The resulting CSI ratio not only provides a reliable, information-bearing phase but also restores the essential complementarity.
Building upon the CSI ratio model, DiverSense~\cite{diversense2022} further enhances noise resilience. The method first constructs multiple CSI ratio streams, which are subsequently aligned in the complex plane and summed to suppress zero-mean Gaussian noise.

However, the state-of-the-art methods for blind-spot-free sensing in~\cite{zeng2018fullbreathe,zeng2019farsense,diversense2022} fundamentally rely on spatial diversity provided by multiple antennas (i.e., by leveraging same-frequency subcarriers from different antenna pairs). The multi-antenna requirement poses a significant challenge for the applications in the growing ecosystem of low-cost Internet of Thing (IoT) devices, as the vast majority are equipped with only a single antenna due to cost and form-factor considerations~\cite{schumann2023WiFi}.
Existing explorations into single-antenna Wi-Fi respiration monitoring have made valuable contributions but also exhibit certain limitations. For instance, some solutions utilize only the CSI amplitude for sensing~\cite{khan2023novel, wang2016human}. While effective under particular conditions, the amplitude-only approaches are inherently susceptible to the blind spot problem.
Other studies recover the CSI phase by calibrating phase offsets, such as the sampling frequency offset (SFO) and the central frequency offset (CFO)~\cite{cai2024device}.
The reliance on calibration, however, necessitates object displacements exceeding one wavelength, a condition that renders the approach unsuitable for monitoring sub-wavelength respiration motions~\cite{zeng2018fullbreathe}.
While distributed deployments, such as Wi-CHAR~\cite{hao2024wi}, can mitigate blind spots by expanding the sensing area, the associated need for complex signal selection algorithms and strict device placements limits practical applicability in residential environments.
\subsection{Contributions}

Motivated by the critical need to eliminate blind spots for ubiquitous low-cost, single-antenna IoT devices, our work investigate the use of frequency diversity.
While state-of-the-art techniques, such as the CSI conjugate multiplication~\cite{zeng2018fullbreathe} and the CSI ratio~\cite{zeng2019farsense,diversense2022}, can eliminate blind spots via spatial diversity, the underlying requirement for multiple antennas renders the approaches unsuitable for such IoT devices.
In parallel, existing solutions developed for such single-antenna devices still cannot fully address the blind-spot issue, thereby compromising practical reliability.
Specifically, our work seeks to address the following key questions:
\begin{itemize}
    \item \textbf{The Blind-Spot Problem:} How can a frequency diversity-based sensing scheme for single-antenna devices effectively cancel hardware-induced random phase offsets, thereby restoring the essential complementarity required to eliminate blind spots?
    \item \textbf{The Signal Quality Problem:} How can the sensing-signal-to-noise ratio (SSNR) of the recovered respiratory signal be optimized to ensure the proposed blind-spot-free sensing scheme is robust and reliable for practical, long-range monitoring?
\end{itemize}

The above questions highlight crucial aspects of designing a robust, blind-spot-free scheme for single-antenna devices. To address the questions, we develop single-antenna-Wi-Fi-sensing (SA-WiSense), a complete signal processing framework for real-time, blind-spot-free and high-SSNR respiration estimation. Our main contributions are follows:
\begin{enumerate}
    \item 
    We propose a novel cross-subcarrier CSI ratio (CSCR) for single-antenna devices, which leverages frequency diversity.
    We establish a fundamental proposition demonstrating that single-target, blind-spot-free sensing is achievable using only a single antenna pair.
    The proof of the proposition is predicated on the elimination of hardware-induced random phase offsets, thereby restoring the inherent complementarity of respiratory signals.
    Our work provides the theoretical foundation for blind-spot-free respiration sensing using single-antenna hardware, overcoming the reliance on spatial diversity from MIMO systems in existing methods~\cite{zeng2019farsense,diversense2022}.
    \item We develop a genetic algorithm (GA) based subcarrier selection (GASS) approach to maximize the SSNR of the proposed CSCR. The approach is founded on our key insight that the strategic selection of subcarriers effectively mitigates position-dependent signal degradation, a challenge previously unaddressed in single-antenna sensing. By formulating and solving the selection as an optimization problem, GASS ensures reliable sensing across both low-noise and high-noise conditions.
    \item We implement our proposed SA-WiSense framework and evaluate the performance using CSI data acquired from commodity ESP32 microcontrollers. Extensive experimental results demonstrate that our framework achieves accurate and full-coverage respiration monitoring, effectively overcoming the blind-spot problem inherent in conventional single-antenna approaches~\cite{khan2023novel, wang2016human,cai2024device}.
\end{enumerate}

The rest of this paper is organized as follows. 
Section~\ref{sec:preliminaries} details the system model and explains the underlying cause of blind spots. 
Section~\ref{sec:theorem} introduces and proves the single-antenna blind-spot-free sensing proposition. 
Section~\ref{sec:system_design} describes the techniques for enhancing the SSNR and presents our framework design of SA-WiSense.
Finally, Section~\ref{sec:EXPERIMENTS AND RESULTS} evaluates the performance of our proposed SA-WiSense, and Section~\ref{sec:conclusion} concludes the paper.
        \section{System Model}
\label{sec:preliminaries}
An ideal CSI model is first presented to establish the concept of complementarity for blind-spot-free detection. Subsequently, we analyze the impact of hardware impairments that corrupt the ideal CSI, revealing the key challenges that motivate our work.

\subsection{Ideal System Model}
\label{subsec:ideal_model}

\begin{figure}[t]
    \centering
    \includegraphics[width=0.7\columnwidth]{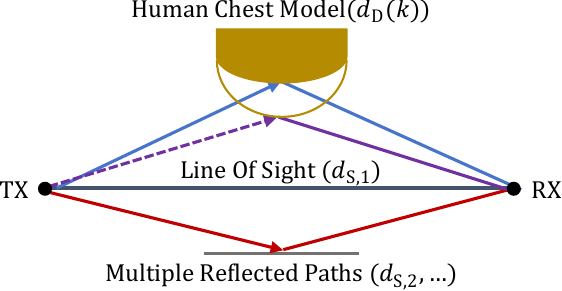}
    \caption{
    The system model for Wi-Fi-based respiration sensing comprises a transmitter (TX), a receiver (RX), and two types of signal paths. The static paths include the line-of-sight (LOS) path and multiple reflected paths, where the length of the $p$-th path is denoted by $d_{\text{S},p}$. The dynamic path has a time-varying length $d_{\text{D}}(k)$ modulated by human respiratory motion at time instant $k$.
    }
    \label{fig:system_model}
\end{figure}

We consider a point-to-point single-input single-output (SISO) Wi-Fi system for single-target respiration sensing, as illustrated in Fig.~\ref{fig:system_model}. The system comprises a single transmitter (TX) and a single receiver (RX), operating under the IEEE 802.11n standard~\cite{wifi802.11} on channel 11 of the 2.4~GHz band with a 40~MHz bandwidth\footnote{The frequency range of the effective subcarriers is from 2433.875~MHz to 2470.125~MHz.}.
The received signal comprises a superposition of signals from multiple propagation paths, consist of static and dynamic types. 
The static paths, including the line-of-sight (LOS) path and reflections from stationary objects, are characterized by fixed path lengths $d_{\text{S},p}$, where $p$ is the path index.
In contrast, the dynamic path is created by the signal reflection from the human chest. The length of the dynamic path, $d_{\text{D}}(k)$, varies over time due to respiratory motion.

The ideal CSI for the $m$-th subcarrier ($m=1,\dots,M$) at the $k$-th time instant ($k=1,\dots,K$) can be modeled as~\cite{wang2016human, zeng2018fullbreathe, zeng2019farsense}:
\begin{equation}
\label{eq:ideal_csi}
\begin{aligned}
H\left( m,k \right) &=H_{\text{S}}\left( m \right) +H_{\text{D}}\left( m,k \right) \\
&=\sum_{p=1}^{L_{\text{S}}}{A_{\text{S},p}\left( m \right) \text{e}^{-\frac{\text{j}2\pi d_{\text{S},p}}{\lambda _m}}}+{A_{\text{D}}\left( m,k \right) \text{e}^{-\frac{\text{j}2\pi d_{\text{D}}\left( k \right)}{\lambda_m}}},
\end{aligned}
\end{equation}
where $M$ is the number of subcarriers\footnote{$M=114$ for High-Throughput Long Training Field (HT-LTF) and $M=104$ for Legacy Long Training Field (L-LTF).}, $K$ is the number of CSI samples.
Fig.~\ref{fig:csi_vector}(a) illustrates the overall CSI given by~\eqref{eq:ideal_csi} in the complex plane. The CSI comprises a static component, $H_{\text{S}}(m)$, and a dynamic component, $H_{\text{D}}(m,k)$.
The static component represents the aggregate effect of all $L_{\text{S}}$ static paths, while the dynamic component results from the signal reflection off the human chest.
$A_{\text{S},p}$ and $A_{\text{D}}(m,k)$ denote the amplitudes of the $p$-th static path and the dynamic path, respectively, and $\lambda_m$ represents the wavelength of the $m$-th subcarrier.
For model reduction, the multiple static paths are aggregated into a single effective path, i.e., $L_{\text{S}}=1$\footnote{The aggregation is justified by the principle of superposition, where the sum of multiple static signal vectors yields a single resultant vector.}.

\subsection{The blind-spot-free Phenomenon}
\label{subsec:blind-spot-free-Phenomenon}
Chest-wall movements induced by respiration cause variations in the dynamic signal path length $d_{\text{D}}(k)$ of approximately 5--12~mm~\cite{wang2016human}.
The resulting path length variations introduce a corresponding change of roughly~$\pi/3$ in the Fresnel phase~$\rho_{\text{fn}}(m,k)$~\cite{wang2016human}, as illustrated in Fig.~\ref{fig:csi_vector}(a). The Fresnel phase is defined as the phase difference between the static and dynamic signal components, i.e., $\rho_{\text{fn}}(m,k) = \angle H_{\text{S}}(m,k) - \angle H_{\text{D}}(m,k)$~\cite{zeng2018fullbreathe}.
Direct measurement of the Fresnel phase, an ideal indicator of respiratory dynamics, is not feasible. Instead, the required variations are inferred from modulations in the accessible CSI amplitude $\left| H(m,k) \right|$ and phase $\theta_{\text{ch}}(m,k)$ for respiration monitoring.

Under ideal noise-free conditions, the CSI amplitude, $\left| H(m,k) \right|$, is given by~\cite{zeng2018fullbreathe}:
\begin{equation}
\label{eq:csi_amplitude}
\left| H(m,k) \right| = \sqrt{
    \begin{multlined}[t]
    \left| H_{\text{S}}(m)\right|^2 + \left| H_{\text{D}}(m,k)\right|^2 \\
    + 2\left| H_{\text{S}}(m) \right|\left| H_{\text{D}}(m,k) \right|\cos \rho _{\text{fn}}\left(m,k\right),
    \end{multlined}
}
\end{equation}
and the CSI phase, $\theta_{\text{ch}}(m,k) = \angle H(m,k)$, is given by:
\begin{equation}
\label{eq:csi_phase_alternative}
\begin{split}
\theta_{\text{ch}}(m,k) = \angle H_{\text{S}}(k) - \frac{\left| H_{\text{D}}(m,k) \right|}{\left| H(m,k) \right|}\sin \rho_{\text{fn}}\left(m,k\right).
\end{split}
\end{equation}
Equations~\eqref{eq:csi_amplitude} and~\eqref{eq:csi_phase_alternative} reveal a cosine dependency for the CSI amplitude and a sine dependency for the CSI phase, both with respect to the Fresnel phase~$\rho_{\text{fn}}(m,k)$.
The orthogonal relationship between the cosine and sine functions establishes a $\pi/2$ phase difference between the ideal CSI amplitude and phase waveforms, a property is known as respiratory signal complementarity~\cite{zeng2018fullbreathe}.

The property of complementarity is fundamental to achieving blind-spot-free sensing.
For instance, when respiratory motion causes the CSI vector to trace the circular arc from~`c' to~`e' in Fig.~\ref{fig:csi_vector}(b), the amplitude $\left| H(m,k) \right|$ exhibits a negligible response. 
A sensing method that relies solely on amplitude would therefore encounter a blind spot.
However, the CSI phase $\theta_{\text{ch}}(m,k)$ experiences a substantial variation over the same trajectory. The inherent quadrature relationship ensures that when the sensitivity of one component to motion is low, the sensitivity of the other is high. Consequently, a joint processing of both CSI amplitude and phase enables robust respiration tracking without blind spots.

\begin{figure}[t]
    \centering
    \begin{subfigure}[b]{0.48\columnwidth}
        \centering
        \includegraphics[width=\linewidth]{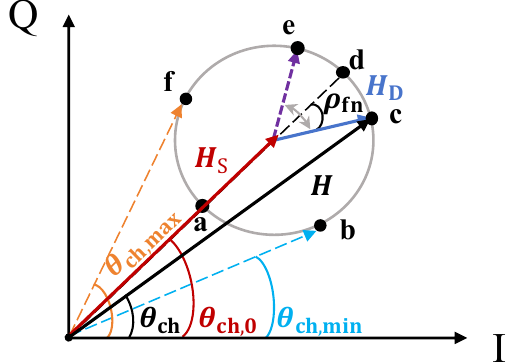}
        \caption{Vector representation of CSI in the complex plane.}
        \label{fig:csi_vector_plane}
    \end{subfigure}
    \hfill
    \begin{subfigure}[b]{0.48\columnwidth}
        \centering
        \includegraphics[width=\linewidth]{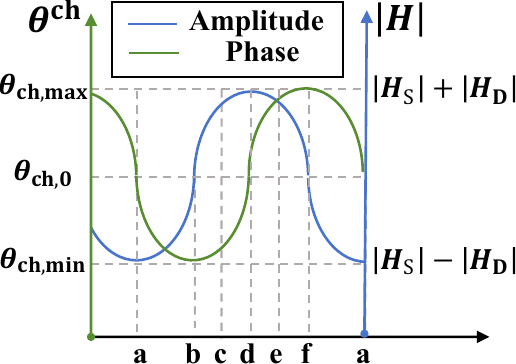}
        \caption{Idealized CSI amplitude and phase waveforms.}
        \label{fig:csi_curves}
    \end{subfigure}
    \caption{Illustration of the ideal CSI model and complementarity. 
    (a)~Vector representation of the ideal CSI in the complex plane, showing the static component $H_{\text{S}}$, the dynamic component $H_{\text{D}}$, and the resultant CSI $H$. Respiratory motion causes $H_{\text{D}}$ to trace a circular arc (e.g., from `c' to `e'). 
    (b)~Idealized CSI amplitude $\left| H \right|$ and phase $\theta_{\text{ch}}$ curves corresponding to the vector movement, illustrating the complementarity through the $\pi/2$ phase difference.}
    \label{fig:csi_vector}
\end{figure}

\subsection{Practical Impairments and the Emergence of Blind Spots}
\label{subsec:practical_model}
In practice, CSI measurements from commodity Wi-Fi devices are corrupted by multiple noise sources. A practical, noisy CSI measurement, $\tilde{H}(m,k)$, is commonly modeled as~\cite{li2021complexbeat}:
\begin{equation}
\label{eq:practical_csi}
\tilde{H}\left( m,k \right) \begin{matrix} =A_{\text{n}}\left( m,k \right) \text{e}^{-\text{j}\theta_{\text{os}}\left( m,k \right)}H\left( m,k \right) +\varepsilon \left( m,k \right)\\ \end{matrix}
\end{equation}
where $A_{\text{n}}(m,k)$ denotes the amplitude impulse noise~\cite{Wu2023WiTraj}, $\theta_{\text{os}}(m,k)$ represents the random phase offset~\cite{li2021complexbeat, Yang2021Respiration}, and $\varepsilon(m,k)$ is the zeros-mean Gaussian noise~\cite{Chen2025WiPhase}.

The key challenge arises from the random phase offsets, $\theta_{\text{os}}(m,k)$.
Unlike the high-frequency amplitude impulse noise and Gaussian noise that can be suppressed via filtering, random phase offsets introduce time-varying and unpredictable distortions to the CSI phase~\cite{Chen2025WiPhase}. 
Such unpredictable variations render the raw phase information unsuitable for direct use in respiration detection.
As a result, sensing methods typically discard the phase component and operate using only the amplitude of $\tilde{H}(m,k)$~\cite{khan2023novel, wang2016human}. Reliance on the amplitude alone, however, is the direct cause of the blind-spot problem, as detailed in Section~\ref{subsec:blind-spot-free-Phenomenon}.

The hardware-induced impairments inherent in commodity Wi-Fi devices, particularly the random phase offsets, amplitude impulse noise and Gaussian noise, lead to two fundamental and interconnected challenges for single-antenna respiration sensing: the blind-spot problem and the signal quality problem. The former arises from the corruption of phase information, making amplitude-only sensing unreliable, while the latter affects the overall robustness of the measurement. In the following sections, we first establish the theoretical foundation for blind-spot-free sensing in Section~\ref{sec:theorem}, and then detail the framework design to enhance signal quality in Section~\ref{sec:system_design}.

        \section{CSCR for Single-Antenna Blind-Spot-Free Sensing}
\label{sec:theorem}
To enable blind-spot-free respiration sensing on single-antenna devices, we propose a novel ratio-based formulation, termed the CSCR, defined as:
\begin{equation}
\label{eq:csi_ratio_model}
\begin{aligned}
\MoveEqLeft \mathcal{H}(m_1,m_2,k) = \frac{\tilde{H}(m_1,k)}{\tilde{H}(m_2,k)} \\&= \frac{A_{\text{n}}\left( m_1,k \right) \text{e}^{-\text{j}\theta _{\text{os}}\left( m_1,k \right)}H\left( m_1,k \right) +\varepsilon \left( m_1,k \right)}{A_{\text{n}}\left( m_2,k \right) \text{e}^{-\text{j}\theta _{\text{os}}\left( m_2,k \right)}H\left( m_2,k \right) +\varepsilon \left( m_2,k \right)}.
\end{aligned}
\end{equation}
where $m_1$ and $m_2$ are two distinct subcarriers ($m_1 \neq m_2$). The proposed CSCR leverages frequency diversity, distinguishing the formulation from conventional CSI ratios~\cite{zeng2019farsense,diversense2022} that depend on spatial diversity from multi-antenna systems.

The proposed CSCR performs respiration monitoring by utilizing both the amplitude $\left| \mathcal{H}(m_1,m_2,k) \right|$ and phase $\vartheta(m_1,m_2,k)$.
To ensure the reliability of the two components, the proposed CSCR formulation is designed to counteract hardware-induced corruptions. 
The formulation first cancels the random phase offsets $\theta_{\text{os}}(m,k)$, a property demonstrated in \textbf{Lemma}~\ref{lem:phase_elimination}.
The ratio also eliminates the amplitude impulse noise $A_{\text{n}}(m,k)$, an operation detailed in \textbf{Lemma}~\ref{lem:amplitude_elimination}. Removal of such amplitude impulse noise enhances overall signal quality for robust sensing.
By nullifying both the phase offsets and the amplitude noise, the CSCR formulation successfully restores the inherent complementarity of the respiratory signal.
The restored complementarity is formally established in \textbf{Proposition}~\ref{thm:blind_spot_free} and is the key to achieving blind-spot-free performance.

\begin{lemma}[Elimination of Random Phase Offsets]
\label{lem:phase_elimination}
The proposed CSCR eliminates random phase offsets and recovers a reliable, information-bearing phase.
\end{lemma}

\begin{proof}
According to the analysis in~\cite{Chen2025WiPhase}, the measured phase of the $m$-th CSI sample, denoted as $\theta_{\text{RX}}(m,k)= \angle \tilde{H}(m,k)$, can be expressed as:
\begin{equation}
\begin{split}
\theta_{\text{RX}}(m,k) &= \theta_{\text{ch}}(m,k) + \theta_{\text{os}}(m,k) \\
          &= \theta_{\text{ch}}(m,k) + n(m)\cdot(\eta_{\text{b}}(k) + \eta_{\text{o}}) + \varphi(k),
\end{split}
\end{equation}
where $\eta_{\text{b}}(k)$, $\eta_{\text{o}}$, and $\varphi(k)$ denote the phase errors originating from packet boundary detection (PBD), SFO, and CFO, respectively, $n(m)$ denotes the physical subcarrier index.
Although $\theta_{\text{ch}}(m,k)$ carries the desired respiratory information, the random phase offset $\theta_{\text{os}}(m,k)$ hinders a direct extraction of the respiratory signal from the composite phase $\theta_{\text{RX}}(m,k)$.
A key property of $\theta_{\text{os}}(m,k)$ is the linear relationship with the physical subcarrier index $n(m)$, which enables the elimination of the offset by computing a ratio between different subcarriers\footnote{Note that the physical index $n(m)$ is distinct from the subcarrier array index $m$For instance, in the IEEE 802.11n HT-LTF field where $n(m) \in [-58, -2] \cup [2, 58]$, the array indices $m \in \{1, 57, 58, 114\}$ correspond to physical indices $n(m=1) = -58, n(m=57) = -2, n(m=58) = 2$, and $n(m=114) = 58$.}. 

The phase noise component $\eta_{\text{b}}(k)$ follows a zero-mean Gaussian distribution~\cite{Chen2025WiPhase}. Therefore, we apply the weak law of large numbers to mitigate the time-varying PBD-induced errors by averaging the measured phase $\theta_{\text{RX}}(m,k)$ over $K_1$ samples. The resulting averaged phase, $\hat\theta_{\text{RX}}(m,k)$, is represented as:
\begin{align}
\MoveEqLeft \hat\theta_{\text{RX}}(m,k) = \sum_{i=1}^{{K_1}} \frac{\theta_{\text{RX}}(m,(k-1){K_1}+i)}{K_1} \\&= \theta_{\text{ch}}(m,k) + n(m)\cdot\eta_{\text{o}} + \varphi(k).
\end{align}
Subsequently, by constructing the proposed CSCR with subcarriers $m_1$ and $m_2$, the CSCR phase, $\vartheta(m_1,m_2,k) = \angle \mathcal{H}(m_1,m_2,k)$, is derived as:
\begin{equation}
\begin{split}
\vartheta&(m_1,m_2,k) 
= \hat{\theta}_{\text{RX}}(m_1,k) - \hat{\theta}_{\text{RX}}(m_2,k) \\ 
&= \theta_{\text{ch}}(m_1,k) - \theta_{\text{ch}}(m_2,k) + (n(m_1) - n(m_2))\eta_{\text{o}}.
\end{split}
\end{equation}
The SFO, $\eta_{\text{o}}$, can be regarded as a constant during the CSI processing period, as the SFO remains stable over several minutes~\cite{Chen2025WiPhase}.
The desired respiratory information is embedded in the differential phase $\Delta\vartheta(m_1,m_2,k) = \Delta\theta_{\text{ch}}(m_1,k) - \Delta\theta_{\text{ch}}(m_2,k)$. 
Although the term $(n(m_1) - n(m_2))\eta_{\text{o}}$ introduces a constant phase offset that varies for different subcarrier pairs, the specific offset does not affect the final respiration measurements.
\end{proof}

\begin{lemma}[Elimination of Amplitude Impulse Noise]
\label{lem:amplitude_elimination}
The proposed CSCR eliminates amplitude impulse noise that exhibits high correlation across different subcarriers.
\end{lemma}
\begin{proof}
As detailed in~\cite{Wu2023WiTraj}, amplitude impulse noise primarily originates from internal state transitions within the Wi-Fi network interface card, such as adjustments to transmission power or rate adaptation. A key characteristic of the impulse noise is the high correlation across all CSI streams, meaning the noise simultaneously corrupts the samples on all subcarriers.  
The high correlation, where $A_{\text{n}}(m_1,k) \approx A_{\text{n}}(m_2,k)$ for $m_1 \neq m_2$, ensures that the correlated noise term $A_{\text{n}}(m,k)$ is present in both the numerator and the denominator during the ratio operation.
Consequently, performing the proposed CSCR operation~\eqref{eq:csi_ratio_model} cancels out the common noise term, which effectively suppresses the amplitude impulse noise.
\end{proof}

\begin{remark}
With the aid of \textbf{Lemma}~\ref{lem:phase_elimination} and \textbf{Lemma}~\ref{lem:amplitude_elimination}, the proposed CSCR expression in~\eqref{eq:csi_ratio_model} reduces to~\eqref{eq:simplified_ratio}:
\begin{equation}
\label{eq:simplified_ratio}
    \mathcal{H}(m_1,m_2,k)
    = \frac{H\left(m_1,k\right)+\dot{\varepsilon}\left( m_1,k \right)}{H\left(m_2,k\right)+\dot{\varepsilon}\left( m_2,k \right)}\text{e}^{-\text{j}\Delta\theta(m_1,m_2)},
\end{equation}
where the term $\dot{\varepsilon}(m_1,k)$ and $\dot{\varepsilon}(m_2,k)$ represents the residual Gaussian noise after being scaled by the impulse noise and rotated by the phase offset. The distribution of $\dot{\varepsilon}(m_1,k)$ and $\dot{\varepsilon}(m_2,k)$ is still a zero-mean Gaussian~\cite{li2021complexbeat}. According to \textbf{Lemma}~\ref{lem:phase_elimination}, the random phase offset is transformed into a time-constant relative phase offset, denoted as $\Delta\theta(m_1,m_2) = (n(m_1) - n(m_2))\eta_{\text{o}}$. 

With the hardware-induced noise effectively eliminated, the remaining challenge is to prove that our proposed CSCR~\eqref{eq:simplified_ratio} restores the essential complementarity.
\end{remark}

\begin{theorem}[Single-Antenna blind-spot-free Sensing]
\label{thm:blind_spot_free}
In a single-target sensing scenario, the proposed CSCR~\eqref{eq:simplified_ratio} enables blind-spot-free sensing by restoring the complementary between the amplitude and phase. The CSCR amplitude $\left| \mathcal{H}(m_1,m_2,k) \right|$ and phase $\vartheta(m_1,m_2,k)$ are given by:
\begin{equation}
\left| \mathcal{H}(m_1,m_2,k) \right| = 
\sqrt{%
    \begin{multlined}[t]
    \left| \mathcal{H}_{\text{S}}(m_1,m_2) \right|^2 + \left| \mathcal{H}_{\text{D}}(m_1,m_2,k) \right|^2 \\
    + 2\left| \mathcal{H}_{\text{S}}(m_1,m_2) \right|\left| \mathcal{H}_{\text{D}}(m_1,m_2,k) \right|\\
    \cos \rho_{\text{ratio}}(m_1,m_2,k),
    \end{multlined}
}
\label{eq:mag_full}
\end{equation}
\begin{equation}
    \vartheta(m_1,m_2,k) \approx 
    \begin{multlined}[t]
    \angle \mathcal{H}_{\text{S}}(m_1,m_2)-\\
    \frac{\left| \mathcal{H}_{\text{D}}(m_1,m_2,k) \right|\sin \rho_{\text{ratio}}(m_1,m_2,k)}{\left| \mathcal{H}(m_1,m_2,k) \right|}.
    \label{eq:phase_full}
    \end{multlined}
\end{equation}
where $\mathcal{H}_{\text{S}}(m_1,m_2)$ and $\mathcal{H}_{\text{D}}(m_1,m_2,k)$ are the static and dynamic components of the CSCR, and $\rho_{\text{ratio}}(m_1,m_2,k) = \angle \mathcal{H}_{\text{S}} - \angle \mathcal{H}_{\text{D}}$ is the Fresnel phase of the CSCR. 
The static and dynamic components are formulated for different noise levels as follows:
For low-noise conditions, $\mathcal{H}_{\text{S}}(m_1,m_2) = \frac{\mathcal{A}}{\mathcal{C}}$ and $\mathcal{H}_{\text{D}}(m_1,m_2,k) = \frac{\mathcal{B}\mathcal{C}-\mathcal{A}\mathcal{D}}{\mathcal{C}^2} \cdot \frac{1}{\mathcal{Z}+\frac{\mathcal{D}}{\mathcal{C}}}$. 
For high-noise conditions, $\mathcal{H}_{\text{S}}(m_1,m_2) = \frac{\mathcal{B}}{\mathcal{D}}$ and $\mathcal{H}_{\text{D}}(m_1,m_2,k) = \frac{\mathcal{A}}{\mathcal{D}} \cdot \mathcal{Z}$.
The terms $\mathcal{A}$, $\mathcal{B}$, $\mathcal{C}$, $\mathcal{D}$, and $\mathcal{Z}$ are defined in~\eqref{eq:def_A}-\eqref{eq:def_Z}.
\end{theorem}

\begin{proof}
The detailed proof is provided in Appendix~\ref{appendix:proof_theorem1}.
\end{proof}

The validity of \textbf{Proposition}~\ref{thm:blind_spot_free} stems from the behavior of the dynamic component $\mathcal{H}_{\text{D}}$. The component $\mathcal{H}_{\text{D}}$ traces a circular arc in the complex plane, a trajectory analogous to that of the dynamic component in ideal CSI, as illustrated in Fig.~\ref{fig:csi_vector}. Consequently, the amplitude and phase relationships in~\eqref{eq:mag_full} and~\eqref{eq:phase_full} become structurally identical to those of the ideal case in~\eqref{eq:csi_amplitude} and~\eqref{eq:csi_phase_alternative}. Restoring such structural identity is crucial for reviving the inherent complementarity between amplitude and phase, thereby eliminating blind spots.
        \section{SA-WiSense Framework Design}
\label{sec:system_design}
Building upon the theoretical feasibility of blind-spot-free sensing established in Section~\ref{sec:theorem}, we now detail the design of SA-WiSense, an algorithm for recovering a high-SSNR respiratory signal. We first formulate a principle for maximizing the SSNR, as detailed in Section~\ref{sec:SSNR_design}. We then present the complete architecture of SA-WiSense algorithm in Section~\ref{sec:system_structure}.

\subsection{Design Principles for SSNR Enhancement}
\label{sec:SSNR_design}
We employ the SSNR to quantify signal quality, which is theoretically defined as $\frac{\left|\mathcal{H}_{\text{D}}(m_1,m_2,k)\right|^2}{\left|\dot{\varepsilon} \left( m_1,k \right)\right|^2}$~\cite{diversense2022}. However, a direct estimation of the noise variance $\left|\dot{\varepsilon} \left( m_1,k \right)\right|^2$ is practically difficult. Therefore, we approximate the SSNR in the frequency domain. Specifically, we first obtain the frequency spectrum of the signal via a fast Fourier transform (FFT). 
The SSNR is then calculated as the ratio of the energy within the respiratory frequency band of 0.167--0.5~Hz to the energy at frequencies above 0.5~Hz. We denote the practical SSNR calculation procedure as the function $R(\cdot)$.

A systematic optimization of the SSNR is contingent upon maximizing the amplitude of the respiratory dynamic component $|\mathcal{H}_{\text{D}}(m_1,m_2,k)|$ while suppressing background noise. 
The enhancement of the desired signal relies on a fundamental principle, optimal signal source selection, which is formally stated in the following lemma.

\newtheorem{lemma3}{Lemma3}
\begin{lemma}[Maximizing Signal Component via Subcarrier Pair Selection]
\label{lem:subcarrier_selection}
The amplitude of the respiratory dynamic component, $|\mathcal{H}_{\text{D}}(m_1,m_2,k)|$, can be maximized by optimally selecting the subcarrier pair $(m_1, m_2)$. The selection is effective in low-noise scenarios for mitigating position-dependent signal degradation and remains beneficial for selecting an optimal signal source in high-noise conditions.
\end{lemma}

\begin{proof}
To demonstrate that strategic subcarrier selection can enhance the respiratory signal, we provide a detailed analysis of the amplitude of the dynamic component, $|\mathcal{H}_{\text{D}}(m_1,m_2,k)|$, under both low-noise and high-noise conditions. Our analysis reveals how proactive selection of the subcarrier pair $(m_1, m_2)$ serves as a fundamental method for improving signal quality.

\paragraph{Low-Noise Scenario}
In a low-noise environment, the behavior of the proposed CSCR is primarily governed by the interaction between static and dynamic signal path components. Based on the Mobius transformation~\eqref{eq:mobius_rewritten} detailed in Appendix~\ref{appendix:proof_theorem1}, the amplitude of the dynamic component, $|\mathcal{H}_{\text{D}}(m_1,m_2,k)|$, can be approximated as:
\begin{equation}
|\mathcal{H}_{\text{D}}(m_1,m_2,k)| \approx \left| \frac{\mathcal{B}\mathcal{C} - \mathcal{A}\mathcal{D}}{\mathcal{C}^2} \cdot \frac{1}{\mathcal{Z} + \frac{\mathcal{D}}{\mathcal{C}}} \right|.
\label{eq:full_mobius_hd}
\end{equation}
Assuming the user's macroscopic position remains stationary over the monitoring period, the propagation distance of the signal path associated with respiration, $d_{\text{D}}(k)$, is considered constant. Consequently, the large-scale fading effects are negligible, which allows us to approximate the dynamic amplitudes $A_{\text{D}}(m_1,k)$ and $A_{\text{D}}(m_2,k)$ as constants independent of the time index $k$. The simplification enables the removal of $k$ from $|\mathcal{H}_{\text{D}}(m_1,m_2,k)|$, leading to the expression:
\begin{equation}
|\mathcal{H}_{\text{D}}(m_1,m_2)| \approx \left| \frac{
\begin{aligned}[b]
& A_{\text{S}}(m_1)A_{\text{D}}(m_2) \text{e}^{-\text{j}2\pi (d_0-d_{\text{S}}) \frac{\lambda_{m_1}-\lambda_{m_2}}{\lambda_{m_1}\lambda_{m_2}}} \\
& \quad - A_{\text{D}}(m_1)A_{\text{S}}(m_2)
\end{aligned}
}{ |A_{\text{S}}(m_2)|^2 - |A_{\text{D}}(m_2)|^2 } \right|,
\label{eq:low_noise_hd_detailed}
\end{equation}
where $d_0-d_{\text{S}}$ is determined by the user's location.
An analysis of \eqref{eq:low_noise_hd_detailed} reveals a potential nulling effect. 
Specifically, for a given subcarrier pair $(m_1, m_2)$, the numerator may approach zero at certain user locations, which are determined by the propagation distance $d_0-d_{\text{S}}$. 
Consequently, the energy of the respiratory signal $|\mathcal{H}_{\text{D}}(m_1,m_2)|$ diminishes to nearly zero, a value too low for reliable measurement.

However, the term, $\text{e}^{-\text{j}2\pi (d_0-d_{\text{S}}) \frac{\lambda_{m_1}-\lambda_{m_2}}{\lambda_{m_1}\lambda_{m_2}}}$, is highly sensitive to the choice of subcarrier wavelengths\footnote{In the 2.4\,GHz band under the 802.11n protocol with a 40\,MHz bandwidth (e.g., channel 11), the term $(\lambda_{m_1}-\lambda_{m_2})/(\lambda_{m_1}\lambda_{m_2})$ can range from -0.12092 to 0.12092. When $d_0-d_{\text{S}} > 4.14$\,m, the exponential term can be tuned to almost any value on the unit circle by simply selecting a different subcarrier pair.}. The tunability provides a powerful mechanism to address position-induced signal nulls.
As a concrete example, consider a simple scenario where $d_0 - d_{\text{S}} = 4$\,m and all channel gain amplitudes are normalized, i.e., $|A_{\text{S}}(m)|=1$ and $|A_{\text{D}}(m)|=0.1$ for any subcarrier $m$.
\begin{itemize}
    \item When we select the subcarrier pair $(m_1=1, m_2=2)$, the calculated amplitude of the dynamic component is $|\mathcal{H}_{\text{D}}(m_1=1,m_2=2)| = 0.00265$.
    \item In contrast, when we select the pair $(m_1=1, m_2=114)$, the amplitude becomes $|\mathcal{H}_{\text{D}}(m_1=1,m_2=114)| = 0.2017$.
\end{itemize}
A 76-fold increase in signal strength was achieved solely by adjusting the subcarrier pair. The improvement highlights that optimizing subcarrier selection is an effective method for mitigating position-dependent signal degradation and enhancing signal quality, particularly in low-noise environments.

\paragraph{High-Noise Scenario}
In a high-noise environment, the amplitude of the dynamic component, $|\mathcal{H}_{\text{D}}(m_1,m_2,k)|$, is given by~\eqref{eq:high_noise_averaged}:
\begin{equation}
    |\mathcal{H}_{\text{D}}(m_1,m_2)| \approx \left| \frac{\mathcal{A}\mathcal{Z}}{\mathcal{D}} \right| = \frac{|A_{\text{D}}(m_1)|}{|A_{\text{S}}(m_2)|}.
    \label{eq:high_noise_hd_simplified}
\end{equation}
Equation~\eqref{eq:high_noise_hd_simplified} indicates that the resulting signal strength is determined by the ratio between the dynamic path gain at subcarrier $m_1$ and the static path gain at subcarrier $m_2$. The value of~\eqref{eq:high_noise_hd_simplified} varies significantly with the chosen subcarrier pair ($m_1, m_2$) due to the inherent frequency selectivity of path gains. Consequently, subcarrier selection remains a beneficial strategy for maximizing the SSNR, even in high-noise environments.

In both low-noise and high-noise scenarios, the subcarrier selection is a crucial mechanism for maximizing the respiratory signal component and enhancing overall system performance.
\end{proof}

The principle in \textbf{Lemma}~\ref{lem:subcarrier_selection} provides a method to address the issue of position-dependent signal degradation, which is critical in both low-noise and high-noise environments.
The probability of identifying an optimal pair $(m_1, m_2)$ can be enhanced through the expansion of the candidate subcarrier pool, which is achieved by leveraging CSI from multiple preamble fields\footnote{For instance, the IEEE 802.11n standard~\cite{wifi802.11} provides CSI from both the L-LTF and the HT-LTF. Utilizing both fields increases the number of available subcarriers from 114 (from HT-LTF alone) to 218.}. Such multi-field approach is robust because a common hardware chain generates the CSI for different fields. Consequently, the hardware-induced noise components between fields are highly correlated, allowing our proposed CSCR approach to effectively cancel random phase offsets and suppress amplitude impulse noise.

\subsection{SA-WiSense Framework}
\label{sec:system_structure}

\begin{figure}[t]
    \centering
    \includegraphics[width=0.95\columnwidth]{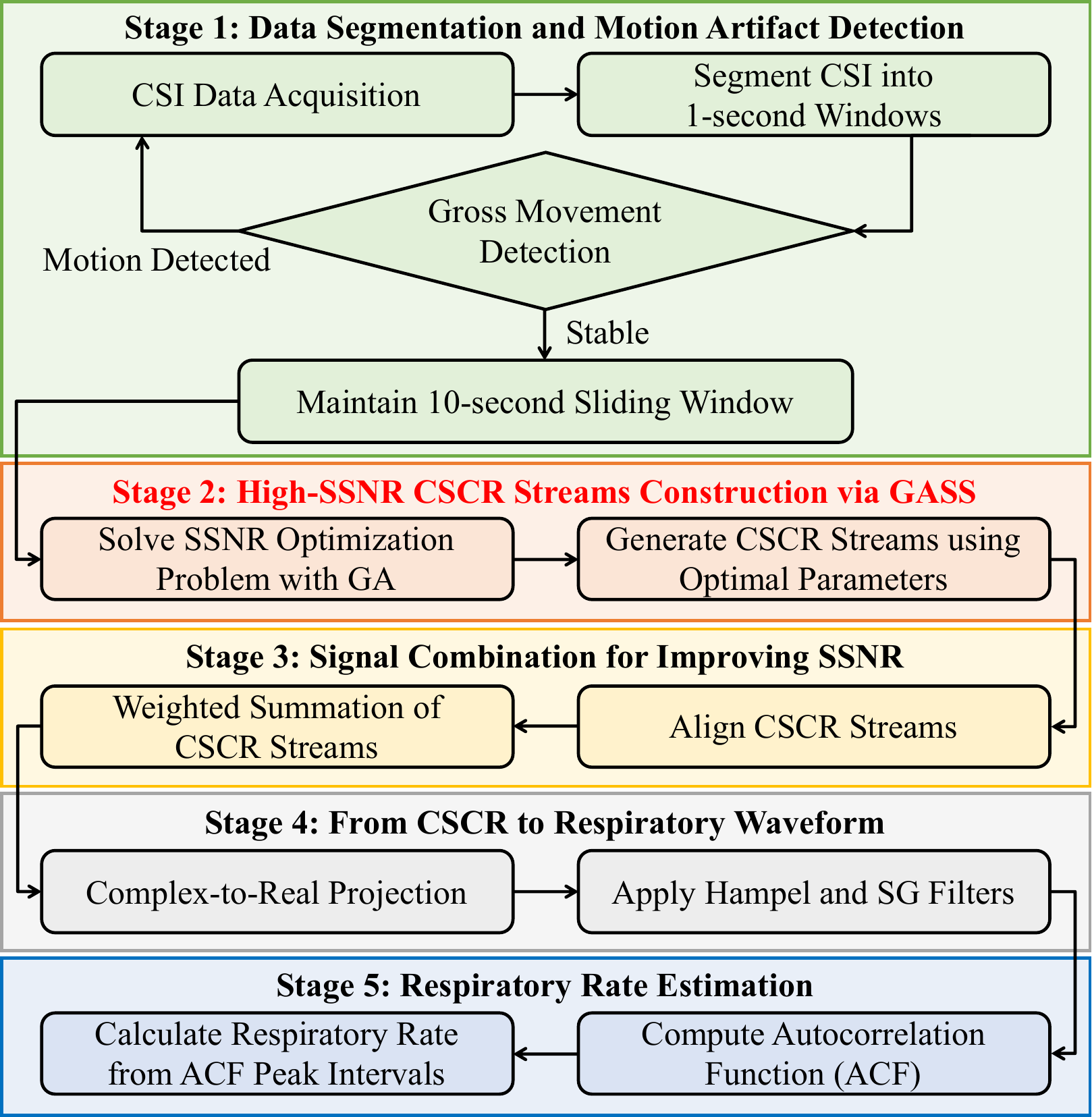}
    \caption{
    The flow diagram of the SA-WiSense framework. Our core innovation in Stage~2 is the proposed GASS approach, which implements our novel CSCR approach to construct high-SSNR data streams for robust sensing, following the initial data segmentation and preceding the subsequent signal combination, waveform recovery, and rate estimation.
    }
    \label{fig:system_diagram_placeholder}
\end{figure}

Based on \textbf{Proposition}~\ref{thm:blind_spot_free} and \textbf{Lemma}~\ref{lem:subcarrier_selection}, we introduce SA-WiSense, a framework for real-time, blind-spot-free and high-SSNR respiratory sensing with single-antenna devices.
As illustrated in Fig.~\ref{fig:system_diagram_placeholder}, the proposed SA-WiSense comprises five sequential stages. 
First, stable 10-second windows are selected from raw CSI to acquire without gross movement.
Next, the core GASS stage optimizes subcarrier selection based on GA and generates high-SSNR CSCR streams.
The CSCR streams are then combined for improving SSNR by suppressing Gaussian noise.
Next, A real-valued respiratory waveform is recovered from the resulting high-SSNR complex signal via complex-to-real projection and refined using Hampel and Savitzky-Golay (SG) filters.
Finally, the respiratory rate is determined by analyzing the peaks of autocorrelation function (ACF) of the refined waveform. 
The implementation details of the stages are presented as follows.

\subsubsection{Data Segmentation and Motion Artifact Detection}
The proposed framework first segments the incoming CSI data stream into 1-second frames for initial motion artifact detection. Adopting the principle from~\cite{bao2023wi}, we identify a frame as containing gross movement if the phase change exceeds a 2-radian threshold, and temporarily pause respiratory signal extraction for such frames.
Frames without gross movement are then concatenated to construct a 10-second sliding window for respiratory rate estimation.
We specifically select the 1-second frame duration to enable near-real-time processing, whereas the 10-second window length is chosen to capture a sufficient number of respiratory cycles (approximately 1.67 to 5 cycles for a typical 0.167--0.5~Hz rate~\cite{wang2016human}) to ensure robust estimation.

\subsubsection{High-SSNR CSCR Streams Construction via GASS}
\label{sec:implementation}
Based on \textbf{Lemma}~\ref{lem:subcarrier_selection} in Section~\ref{sec:SSNR_design}, selecting optimal subcarrier pairs to construct a CSCR improves the resulting SSNR.
With 218 available subcarriers, the search space of potential combinations is immense, rendering an exhaustive search computationally prohibitive.
We therefore formulate the selection task as an optimization problem and employ a GA, a powerful heuristic method, well-suited for efficiently finding near-optimal solutions in the large-scale search spaces.

A combined CSCR is constructed as a weighted sum of $N$ CSCRs, i.e., $\sum_{i=1}^N{a_i\mathcal{H}(m_i,m_{\text{d}},k)}$, where all component CSCRs share a common denominator subcarrier index $m_{\text{d}}$. 
To maximize the SSNR of the combined CSCR, we determine the optimal weighting coefficients $\{a_i|i=1\cdots N\}$ and subcarrier indices $\{m_i|i=1\cdots N, m_{\text{d}}\}$ by solving the following optimization problem: 
\begin{equation}
\textbf{P1:} \underset{ \{a_i,m_i|i=1\cdots N\},m_{\text{d}}}{\text{max}} \quad R\left(\displaystyle\sum_{i=1}^N{a_i\mathcal{H}(m_i,m_{\text{d}},k)}\right),
\end{equation}
\begin{equation}
\text{subject to } \left\{
\begin{array}{@{}l@{\thinspace}l}
|a_i| \le 1, \quad i=1,2,...,N, \\
1 \le m_i \le M, \quad m_i \in \mathbb{Z}, \\
1 \le m_{\text{d}} \le M, \quad m_{\text{d}} \in \mathbb{Z}.
\end{array}
\right.,
\end{equation}
where $\mathbb{Z}$ represents the set of integers. 
Each individual CSCR term in the summation eliminates random phase offsets and impulse noise, while restoring the complementary characteristics of amplitude and phase.
We employ a GA to solve the optimization problem. The elite-preserving strategy of the GA ensures that the SSNR of the resulting combined CSCR is no less than that of any component CSCR within the combination.

Upon convergence, the GA yields the optimal parameters $\{a_i^*, m_i^*,|i=1\cdots N\}$, which are then used to construct a set of high-SSNR CSCR streams for the subsequent signal combination stage, calculated as follows:
\begin{equation}
\mathcal{H}_{\text{str}}(m,k)=\frac{\displaystyle\sum_{i=1}^N{a_i^* \tilde{H}\left( m_i^*,k \right)}}{\tilde{H}\left( m,k \right)}.
\end{equation}

\subsubsection{CSCR streams Combination for Improving SSNR}
To further improve the SSNR, we employ the signal combination technique from~\cite{diversense2022}. The principle of the signal combination implementation lies in the coherent combination of multiple CSCR streams obtained by Section~\ref{sec:implementation}. Such a combination process effectively averages out uncorrelated noise while constructively aggregating in-phase respiratory signals, thereby enhancing the overall signal SSNR.

Coherent combination relies on aligning the CSCR streams $\mathcal{H}_{\text{str}}(m,k)$ in the complex plane. Such an alignment operation compensates for the distinct initial phase and amplitude of each stream. The alignment process, detailed from~\eqref{Qmk} to~\eqref{theta_mk}, initiates by subtracting the amplitude offset from each stream. The result after the subtraction, $Q(m,k)$, is given by~\cite{diversense2022}:
\begin{equation}
\label{Qmk}
Q(m,k) = \mathcal{H}_{\text{str}}(m,k) - \sum_{i=1}^K{\frac{\mathcal{H}_{\text{str}}(m,k)(m,(k-1)K+i)}{K}}.
\end{equation}
Next, we calculate the average amplitude $G(m)$ for each stream:
\begin{equation}
\label{Gm}
G(m) = \max_k\left(\left| \frac{1}{K_3}\sum_{i=k}^{k+K_3-1}{Q(m,i)} \right|\right),
\end{equation}
where $K_3$ is the window size for calculating the average. 
Each stream $Q(m,k)$ is then normalized by its average amplitude $G(m)$ to produce the normalized signal $V(m,k)$, given by:
\begin{equation}
\label{Vmk}
    V(m,k) = G(m) Q(m,k).
\end{equation}
Next, the rotation angle $\varTheta_m$ is calculated for each signal $V(m,k)$ relative to the reference signal $V(m_0,k)$:
\begin{equation}
\label{theta_mk}
\underset{\varTheta_{m}}{\arg\min}\ \text{Dis}(m) = \sum_{k=1}^K{\left| V(m_0,k) - V(m,k)\text{e}^{\text{j}\varTheta_m} \right|^2}.
\end{equation}
The alignment is completed by rotating each signal $V(m,k)$ by the angle $\varTheta_m$. 

After alignment, the signals are combined using a weighted sum, which is determined by the SSNR. We first calculate an initial weight $\beta(m) = R(\mathcal{H}_{\text{str}}(m,k))$. To filter out noisy streams, a threshold $\mu \beta_0$ is applied, where $\beta_0 = \max(\beta(m))$ and $\mu \in [0,1]$. The final weight $\gamma(m)$ is given by $\gamma(m) = \beta(m) u(\beta(m) - \mu\beta_0)$, where $u(\cdot)$ is the Heaviside step function. The aligned CSCR signals are then combined with the weights to produce the combined signal $\mathcal{H}_{\text{Fd}}(k)$:
\begin{equation}
\label{HFd}
\mathcal{H}_{\text{Fd}}(k) = \sum_{m=1}^M{\left(\gamma(m)V(m,k)\text{e}^{\text{j}\varTheta_m}\right)}.
\end{equation}
Finally, the signal is smoothed using a moving average filter with a 0.33-second window ($I=\lfloor 0.33F_\text{s} \rfloor$ samples, where $F_\text{s}$ denotes the sampling rate)~\cite{diversense2022} to yield the final enhanced signal $\mathcal{H}_{\text{FTd}}(k)$:
\begin{equation}
\label{HFTd}
\mathcal{H}_{\text{FTd}}(k) = \frac{1}{I}\sum_{i=1}^I{\mathcal{H}_{\text{Fd}}((k-1)I+i)}.
\end{equation} 
 
\subsubsection{From CSCR to Respiratory Waveform}
The GA and signal combination processing yield an SSNR-enhanced complex signal, $\mathcal{H}_{\text{FTd}}(k)$. 
To generate a real-valued waveform for rate calculation, $\mathcal{H}_{\text{FTd}}(k)$ is projected onto an optimal axis in the complex plane.
The projection axis, defined by an angle $\varTheta_\text{proj}$, is optimized to maximize the SSNR of the resulting real-valued signal. As proposed in~\cite{zeng2019farsense}, the optimal projected signal, $\mathcal{H}_{\text{proj}}(k)$, is obtained by:
\begin{multline}
\mathcal{H}_{\text{proj}}(k) = \max_{0\le \varTheta_\text{proj} < 2\pi} \Big( R \big( \cos\varTheta_\text{proj} \cdot \operatorname{Re}(\mathcal{H}_{\text{FTd}}(k)) \\
+ \sin\varTheta_\text{proj} \cdot \operatorname{Im}(\mathcal{H}_{\text{FTd}}(k)) \big) \Big),
\end{multline}
where $\operatorname{Re}(\mathcal{H}_{\text{FTd}}(k))$ and $\operatorname{Im}(\mathcal{H}_{\text{FTd}}(k))$ are the real and imaginary parts of the signal, respectively.
The chosen projection scheme leverages the complementarity between the amplitude and phase of $\mathcal{H}_{\text{FTd}}(k)$ to enable blind-spot-free sensing.

To obtain the final, clean respiratory waveform, we apply two filters to $\mathcal{H}_{\text{proj}}(k)$. First, a Hampel filter removes sporadic impulse noise not fully eliminated by the proposed CSCR. Second, a SG filter smooths the waveform, effectively suppressing high-frequency noise while preserving the true morphological features of the respiratory signal~\cite{Tan2022LSTformer}. The resulting filtered signal, $\mathcal{H}_{\text{filt}}(k)$, is then used for respiratory information extraction.

\begin{figure}[t]
    \centering
    \includegraphics[width=0.85\columnwidth]{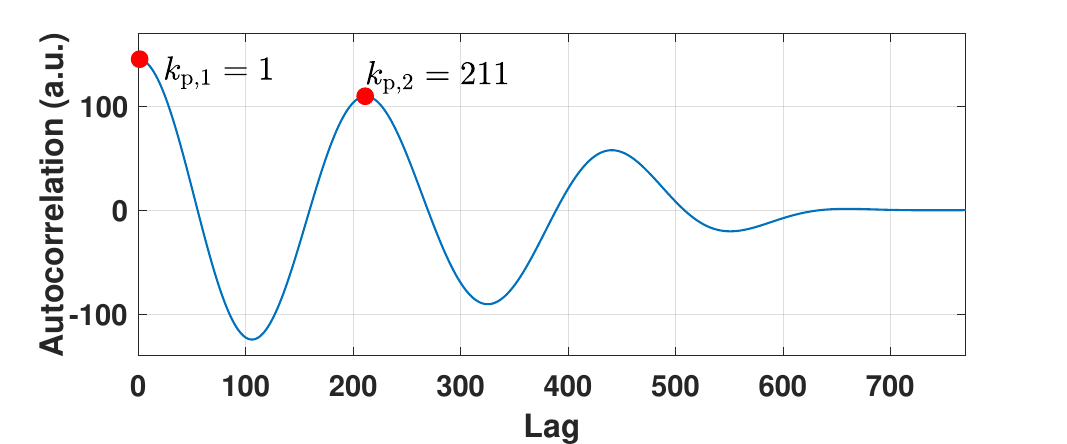}
    \caption{An example of the autocorrelation function (ACF) computed from a processed respiratory waveform ($\mathcal{H}_{\text{filt}}$). The red points are the indices of the peaks of the ACF. The respiratory rate is calculated based on the lag (the number of samples) between the first two consecutive peaks of the ACF, denoted as $k_{\text{p},1}$ and $k_{\text{p},2}$.}
    \label{fig:autocorrelation_example}
\end{figure}

\subsubsection{Respiratory Rate Estimation}
The respiratory rate is calculated from the filtered waveform, $H_{\text{filt}}(k)$, using an autocorrelation-based method.
First, we compute the ACF of the signal. Then, the respiratory rate, denoted as $f_{\text{bpm}}$ in breaths per minute (bpm), is calculated by:
\begin{equation}
    f_{\text{bpm}} = 60 \cdot \frac{F_\text{s}}{|k_{\text{p},2} - k_{\text{p},1}|},
\end{equation}
where $k_{\text{p},1}$ and $k_{\text{p},2}$ represent the indices of the first and second ACF peaks, respectively. As illustrated in Fig.~\ref{fig:autocorrelation_example}, for a signal with a 50~Hz sampling frequency, peaks at indices 1 and 211 yield a respiratory rate of $f_{\text{bpm}} = 60 \cdot {50/|211-1|} = 14.29 \text{bpm}$. To ensure accurate peak identification in the noisy and quasi-periodic respiration signal, we employ the automatic peak detection algorithm from~\cite{John2022Multimodal}.
        \section{Experimental Validation}
\label{sec:EXPERIMENTS AND RESULTS}
A series of experiments is performed to validate the effectiveness of the proposed SA-WiSense framework. The validation is structured into three stages. First, we demonstrate the capability of the proposed CSCR to eliminate hardware-induced random phase offsets and amplitude impulse noise. Second, a simulation is conducted to verify the ability of the proposed SA-WiSense to achieve blind-spot-free respiration sensing by restoring the inherent complementarity of the respiratory signal. Third, the performance of the framework is evaluated in practical scenarios with targets at various locations and distances.

\subsection{Experiment Set-Up}
\subsubsection{Hardware and Data Acquisition}
We use a commodity ESP32 microcontroller to acquire CSI data. The device operates on channel 11 of the 2.4~GHz band with a 40~MHz bandwidth, compliant with the IEEE 802.11n standard~\cite{wifi802.11}. CSI data are collected via the Wi-ESP CSI tool~\cite{wiesp_tool}, which extracts measurements from both the LTF and the HT-LTF, yielding 218 usable subcarriers.
The CSI sampling rate $F_\text{s}$ is configured to 120~Hz, and we verify the precise rate using the `local\_timestamp' provided by the tool.
Five targets participate in the experiments.
In addition, the HKH-11Cplus respiration monitor sensor is uesd to measured the ground truths. During data acquisition, each target remains in a seated position and stay still to minimize motion artifacts.

\subsubsection{Simulation Environment}
We develop a simulation environment to validate the blind-spot-free sensing capability of SA-WiSense. In the simulation, we model a virtual target that emulates a respiration-induced chest movement with a 6~mm depth. The target is positioned at various locations along the LOS path between the TX and RX. 
The corresponding CSI signals are generated based on the model in~\eqref{eq:practical_csi}. The generated CSI data are then used to compute the proposed CSCR. The resulting amplitude, phase, and final projected output of the CSCR are analyzed to evaluate the sensing performance.

\subsection{Validation of Noise Cancellation}
\subsubsection{Hardware Noise Cancellation}
Fig.~\ref{fig:cancellation_demonstration} illustrates the noise cancellation performance of our proposed CSCR approach via an ESP32 platform.
The raw subcarrier phase, shown in Fig.~\ref{fig:cancellation_demonstration}(a), is dominated by large fluctuations from random phase offsets, rendering the phase unsuitable for sensing.
In contrast, the phase of our CSCR, presented in Fig.~\ref{fig:cancellation_demonstration}(c), effectively removes the fluctuations and extracts the respiratory waveform.
In addition, our CSCR mitigates the severe impulse noise corrupting the raw amplitude in Fig.~\ref{fig:cancellation_demonstration}(b), as demonstrated by the cleaned signal in Fig.~\ref{fig:cancellation_demonstration}(d).
While minor residual artifacts are occasionally present, the resulting amplitude is sufficiently clean for reliable respiration analysis.

\begin{figure}[t]
    \centering
    \begin{subfigure}[b]{0.48\columnwidth}
        \centering
        \includegraphics[width=\textwidth]{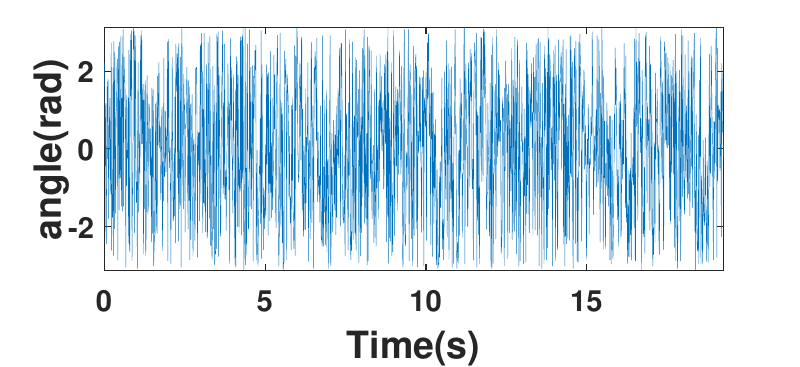}
        \caption{Raw Phase of Subcarrier 50.}
        \label{fig:phase_b}
    \end{subfigure}
    \hfill
    \begin{subfigure}[b]{0.48\columnwidth}
        \centering
        \includegraphics[width=\textwidth]{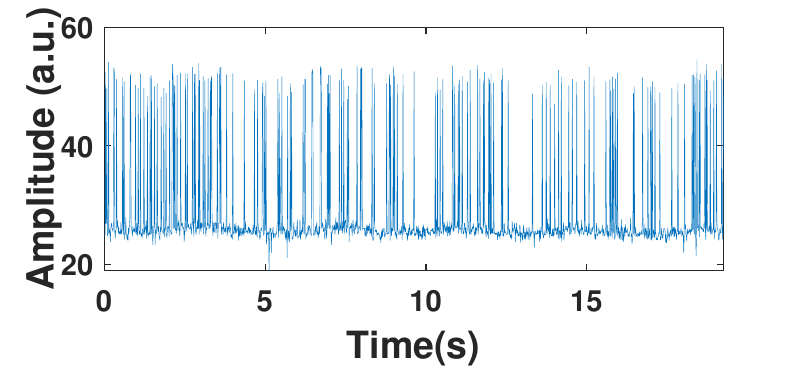}
        \caption{Raw Amplitude of Subcarrier 56.}
        \label{fig:amp_b}
    \end{subfigure}

    \vspace{0.1cm} % Adds a small vertical space between rows

    \begin{subfigure}[b]{0.48\columnwidth}
        \centering
        \includegraphics[width=\textwidth]{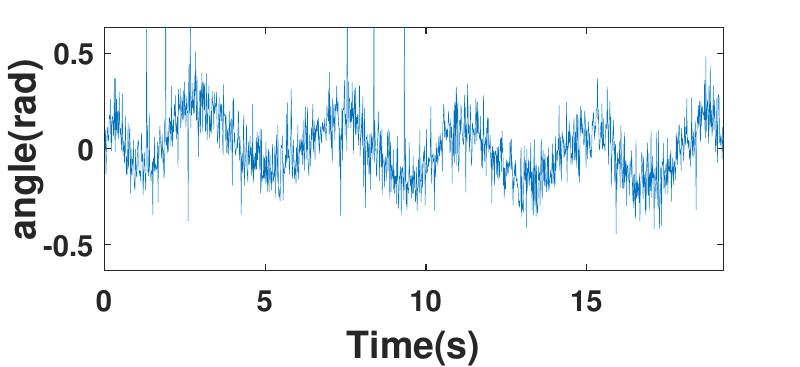}
        \caption{Phase of the Proposed CSCR.}
        \label{fig:phase_c}
    \end{subfigure}
    \hfill
    \begin{subfigure}[b]{0.48\columnwidth}
        \centering
        \includegraphics[width=\textwidth]{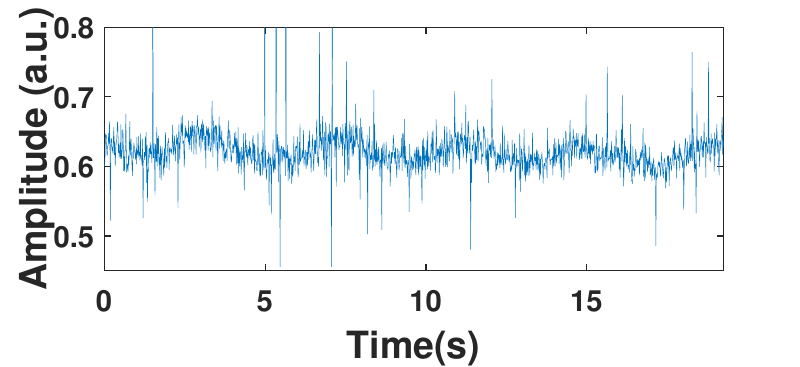}
        \caption{Amplitude of the Proposed CSCR.}
        \label{fig:amp_c}
    \end{subfigure}
    \caption{
    Demonstration of random phase offset cancellation (left column) and amplitude impulse noise suppression (right column). (a) presents the raw CSI phases from a subcarrier, corrupted by random phase offsets. (c) shows the information-bearing phase of the proposed CSCR. (b) displays the raw CSI amplitudes corrupted by impulse noise. (f) displays that the proposed CSCR effectively suppresses such impulse noise.
    }
    \label{fig:cancellation_demonstration}
\end{figure}

\subsubsection{blind-spot-free Sensing Simulation}
We perform a simulation to validate the capability of our proposed CSCR in restoring signal complementarity for blind-spot-free sensing.
Fig.~\ref{fig:blind_spot_verification}(a) presents the amplitude and phase of the processed CSCR at various target-to-LOS distances, which are generated from a simulated target at various distances.
The plot clearly shows that at certain locations, such as 49.8~cm and 53.7~cm, the CSCR amplitude exhibits minimal variation, indicating the insensitivity to the respiratory motion. 
Conversely, at the exact same locations, the phase component exhibits high sensitivity, effectively capturing the motion.
The restored complementarity is fundamental to our approach.
By leveraging the property, the proposed SA-WiSense framework projects the CSCR to reconstruct a continuous respiratory waveform, as depicted in Fig.~\ref{fig:blind_spot_verification}(b).
The result confirms that our proposed SA-WiSense achieves robust and full-coverage respiration monitoring by successfully eliminating the blind spots.

\begin{figure}[t]
    \centering
    \begin{subfigure}[b]{\columnwidth}
        \centering
        \includegraphics[width=\textwidth]{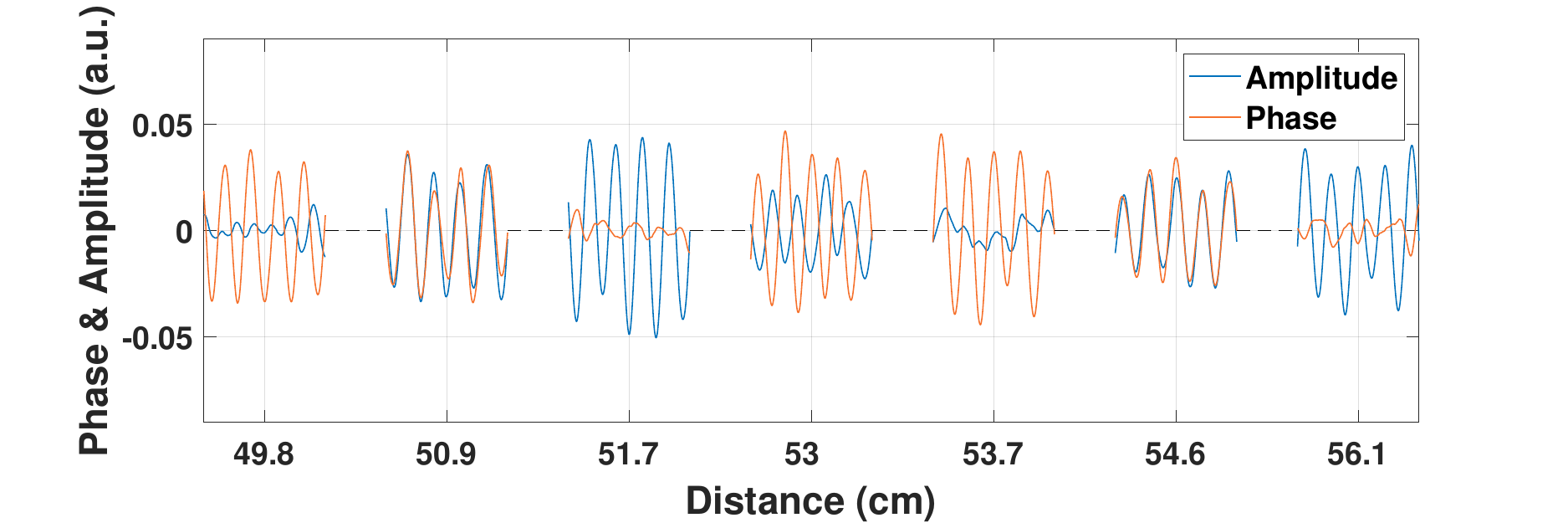}
        \caption{
        Amplitude and phase of the proposed CSCR.
        }
        \label{fig:blind_spot_a}
    \end{subfigure}
    \\
    \vspace{0.3cm}
    \begin{subfigure}[b]{\columnwidth}
        \centering
        \includegraphics[width=\textwidth]{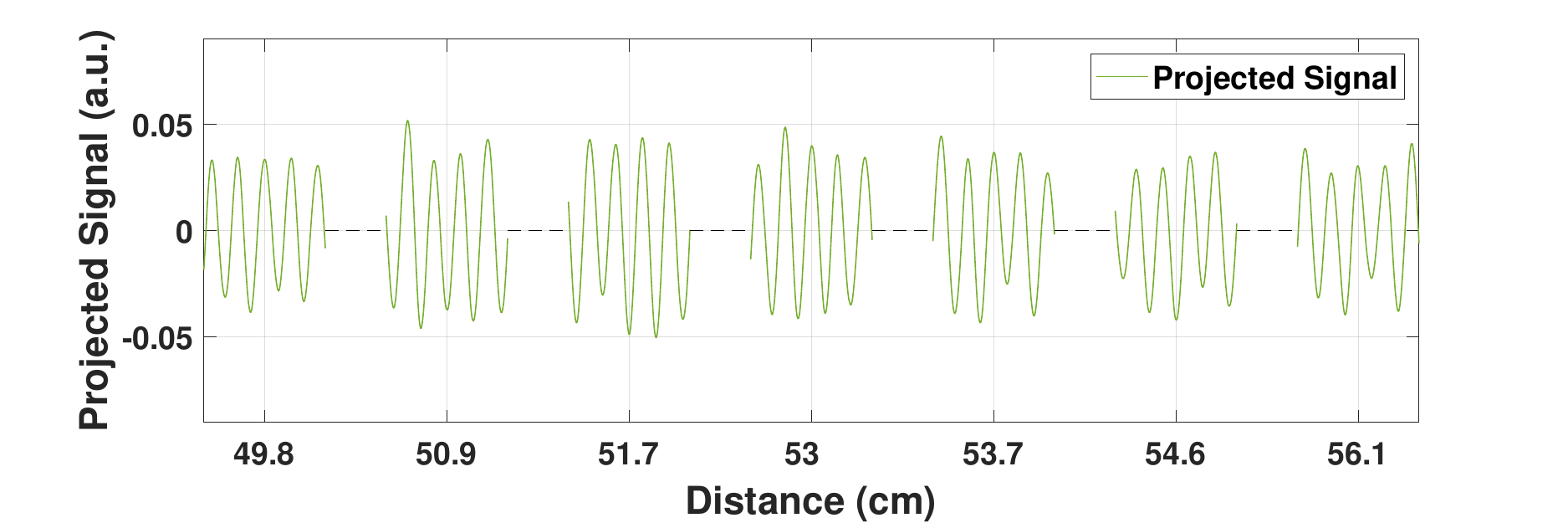}
        \caption{Projected signal from the proposed SA-WiSense.}
        \label{fig:blind_spot_b}
    \end{subfigure}
    \caption{
    Demonstration of blind-spot-free sensing for a human target at various distances along the LOS path. (a) shows complementarity of the proposed CSCR amplitude and phase. The amplitude (blue) is negligible at certain locations (e.g., 49.8~cm and 53.7~cm) where the phase (orange) is sensitive to motion. Conversely, the phase shows minimal response at other locations (e.g., 51.7~cm and 56.1~cm) where the amplitude remains sensitive. (b) shows the signal projected by the proposed SA-WiSense, which leverages both the amplitude and phase to ensure robust sensing performance across all target locations.
    }
    \label{fig:blind_spot_verification}
\end{figure}

\subsection{Impact of Target's Location on Sensing Performance}

\begin{figure}[!t]
    \centering
    \begin{subfigure}[b]{0.48\columnwidth}
        \centering
        \includegraphics[width=\textwidth]{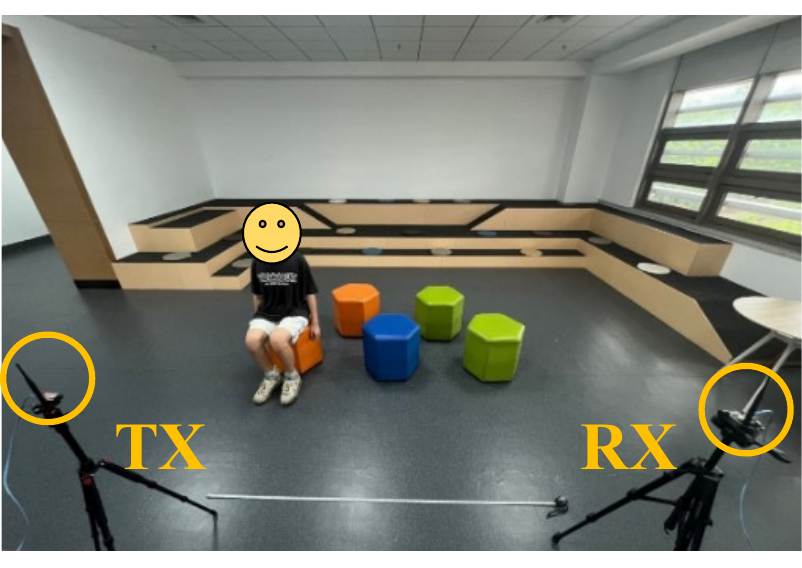}
        \caption{Experimental environment layout.}
        \label{fig:exp_env}
    \end{subfigure}
    \hfill
    \begin{subfigure}[b]{0.48\columnwidth}
        \centering
        \includegraphics[width=\textwidth]{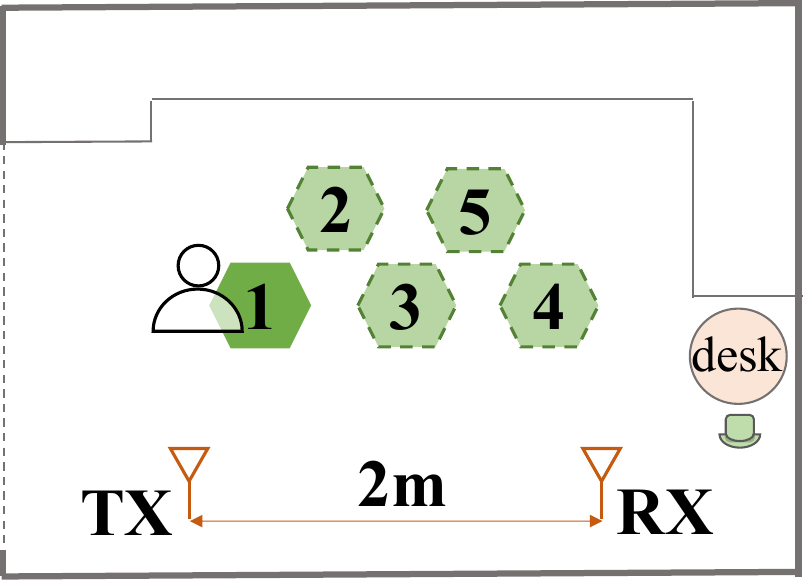}
        \caption{Five sensing locations.}
        \label{fig:exp_pos}
    \end{subfigure}
    \caption{Experimental setup for evaluating the impact of locations. The distance between the TX and RX is 2~m. The target is sequentially positioned at five distinct locations marked 1 through 5.}
    \label{fig:exp_locations}
\end{figure}
\begin{figure}[!t]
	\centering
	\includegraphics[width=\columnwidth]{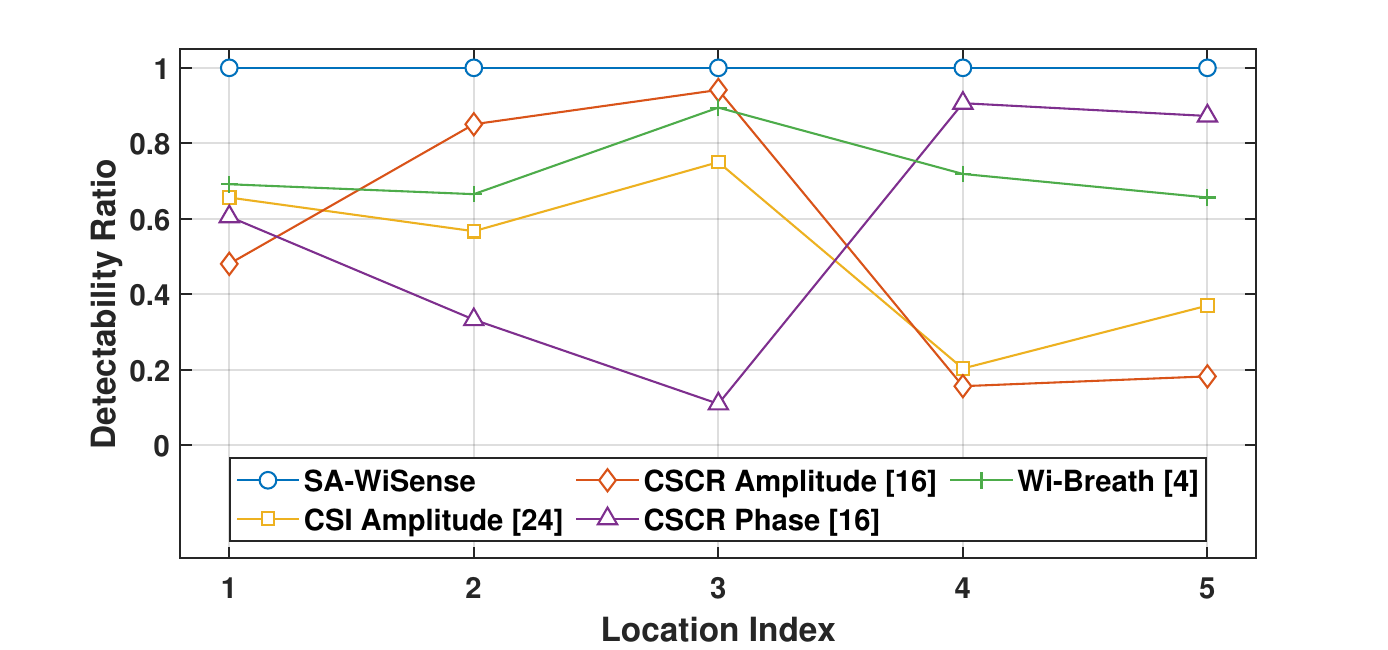}
	\caption{Comparison of detectability ratio for different sensing schemes at five locations. SA-WiSense maintains a 100\% detectability ratio across all locations, demonstrating the effectiveness in eliminating blind spots.}
	\label{fig:loc_results}
\end{figure}
Fig.~\ref{fig:loc_results} shows the performance evaluation of our proposed SA-WiSense against blind spots, measured by the detectability ratio~\cite{zeng2018fullbreathe}.
The experiment is conducted in the indoor environment depicted in Fig.~\ref{fig:exp_locations}, where a target remained stationary at five distinct locations to simulate diverse initial Fresnel phases. 
Our proposed SA-WiSense achieves a 100\% detectability ratio across all five locations, demonstrating robust respiration monitoring regardless of the target's position.
The underlying mechanism for such robust performance is the effective fusion of CSCR amplitude and phase, a principle validated by comparing with component-wise baselines (CSCR amplitude and CSCR phase) adopted from the evaluation methodology in~\cite{zeng2018fullbreathe}.
For instance, the CSCR amplitude method succeeds at location 3 where the CSCR phase method fails, while the opposite occurs at location 4.
By adaptively projecting the complex CSCR as detailed in Section~\ref{sec:system_structure}, SA-WiSense effectively achieves full-coverage sensing by leveraging the inherent complementarity between the amplitude and phase, a mechanism theoretically grounded in \textbf{Proposition}~\ref{thm:blind_spot_free}.
In contrast, the conventional CSI amplitude method~\cite{khan2023novel} suffers from significant performance degradation at locations 4 and 5, confirming the existence of location-dependent blind spots in amplitude-only schemes. In addition, SA-WiSense demonstrates superior performance compared to the state-of-the-art Wi-Breath~\cite{bao2023wi}. The performance of Wi-Breath is constrained in our single-antenna setup, as the scheme was originally designed to leverage spatial diversity from multi-antenna systems.

\subsection{Impact of Target's Distance on Sensing Performance}

\begin{figure}[!t]
    \centering
    \begin{subfigure}[b]{0.48\columnwidth}
        \centering
        \includegraphics[width=\textwidth]{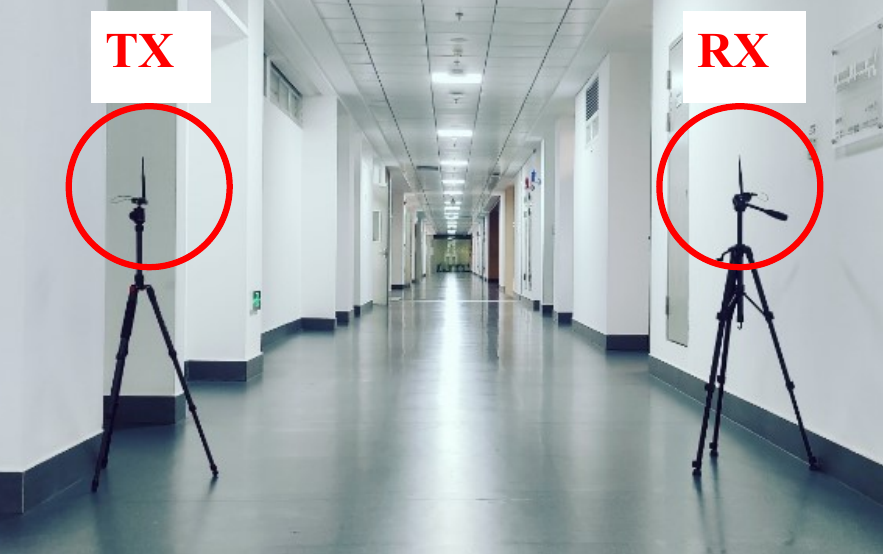}
        \caption{Physical experimental scene.}
        \label{fig:exp2_env}
    \end{subfigure}
    \hfill
    \begin{subfigure}[b]{0.48\columnwidth}
        \centering
        \includegraphics[width=\textwidth]{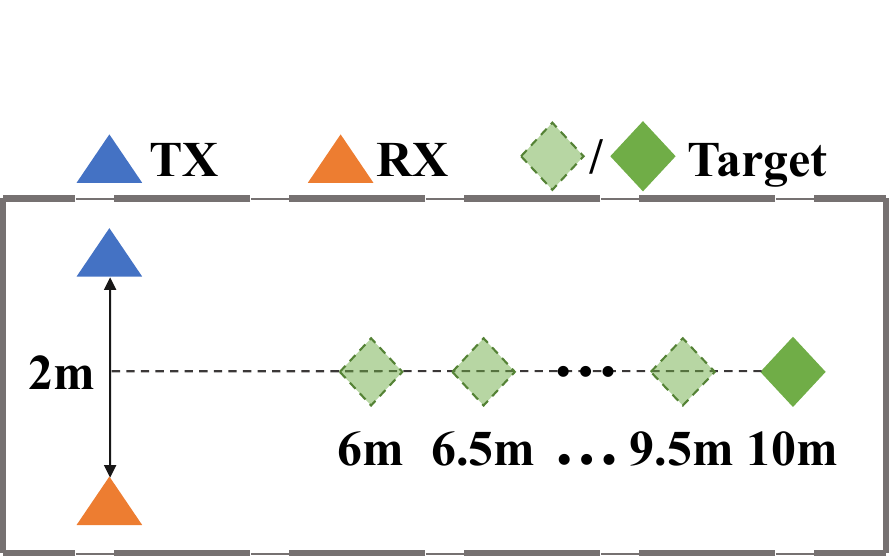}
        \caption{Geometric layout.}
        \label{fig:exp2_dis}
    \end{subfigure}
    \caption{Experimental setup for evaluating the impact of sensing distance. The TX and RX are placed 2 m apart, and the target's distance from the LOS path is varied from 6~m to 10~m.}
    \label{fig:exp2_distances}
\end{figure}
\begin{figure}[!t]
	\centering
	\includegraphics[width=\columnwidth]{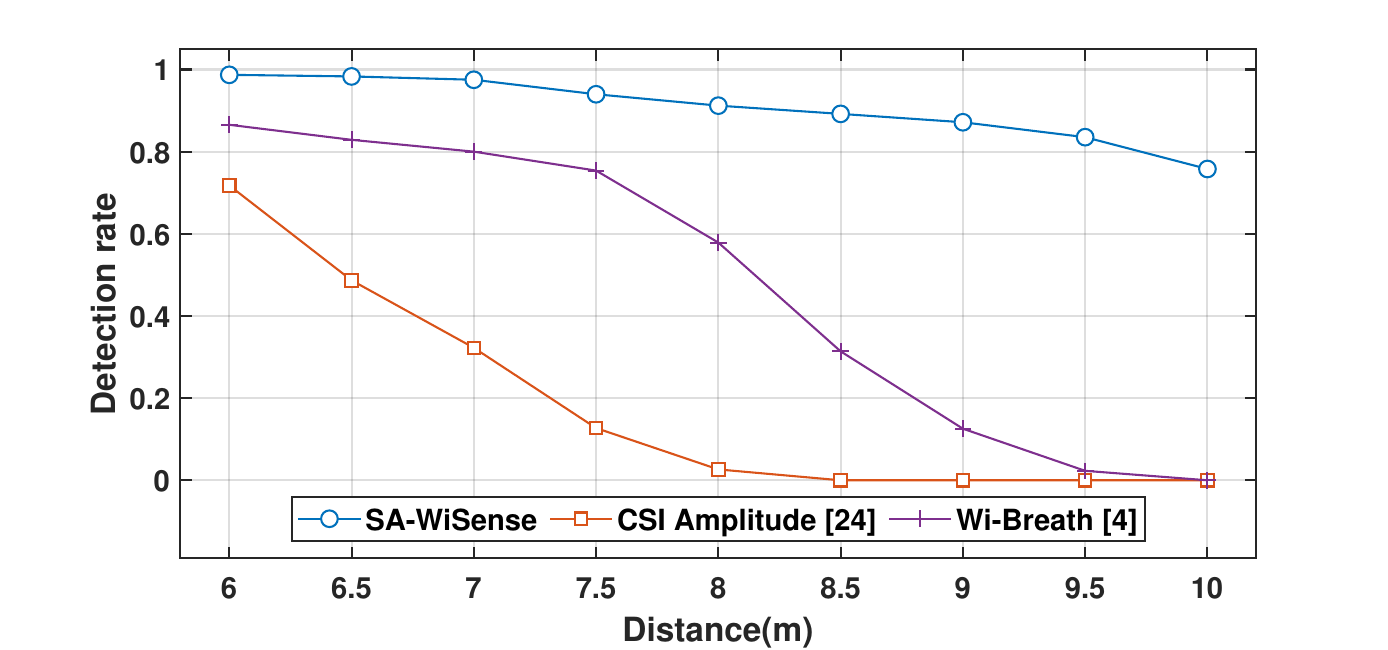}
	\caption{Comparison of detection rate for SA-WiSense and baseline schemes at different sensing distances.}
	\label{fig:dis_results}
\end{figure}

Fig.~\ref{fig:dis_results} presents a performance comparison of SA-WiSense against two baseline schemes over extended sensing distances. 
The experimental environment is depicted in Fig.~\ref{fig:exp2_distances}, where a target is positioned at distances from 6~m to 10~m away from the LOS path.
For evaluation, the detection rate is defined as the percentage of CSI measurements with an absolute estimation error of less than 1~bpm~\cite{zeng2019farsense}.
The results show that all schemes perform well at the initial distance of 6~m. As the distance increases to 8~m, the detection rate of the CSI Amplitude scheme~\cite{khan2023novel} falls to nearly zero, and the rate for Wi-Breath~\cite{bao2023wi} decreases to below 60\%. By contrast, SA-WiSense sustains a high detection rate of over 91.2\% at all tested distances up to 8~m. Therefore, our proposed SA-WiSense enables robust sensing for single-antenna devices, especially under low-SSNR conditions caused by extended distances.
        \section{Conclusion}
\label{sec:conclusion}
This paper proposed SA-WiSense, a cost-effective framework for robust human respiration monitoring using single-antenna Wi-Fi devices. The framework addresses the critical blind-spot problem by introducing a novel CSCR approach. We proved that the proposed CSCR leverages frequency diversity to cancel hardware-induced random phase offsets, thereby restoring the signal complementarity essential for full-coverage sensing.
To further ensure robust performance, we developed a GASS approach. The approach formulates an optimization problem to maximize the SSNR of the constructed CSCR.
We implemented the SA-WiSense framework on commodity ESP32 microcontrollers.
Experimental results demonstrate that SA-WiSense achieves a 100\% detectability ratio across all tested target locations, effectively overcoming the location-dependent blind spot issue.
Furthermore, SA-WiSense framework maintains a high detection rate of 91.2\% for distances up to 8.0 meters, which significantly outperforms conventional single-antenna approaches.
The current work focuses on single-target scenarios. Future work will extend the SA-WiSense framework to support simultaneous multi-person respiration monitoring.

        \vspace{-1mm}
\appendices

\color{black}
\section{Proof of Proposition~\ref{thm:blind_spot_free}}
\label{appendix:proof_theorem1}
We prove \textbf{Proposition}~\ref{thm:blind_spot_free} by demonstrating that the proposed CSCR restores the complementarity of the amplitude and phase under both low-noise and high-noise conditions.

\subsection{Low-Noise Condition}
\label{Low-Noise Condition}
Under the low-noise condition, the influence of the residual Gaussian noise $\dot{\varepsilon}(m,k)$ is negligible ($A_{\text{D}}(m,k) \gg \dot{\varepsilon}(m,k)$). The proposed CSCR from \eqref{eq:simplified_ratio} is approximated as:
\begin{equation}
\label{eq:low_noise_ratio}
\begin{aligned}
\MoveEqLeft \mathcal{H}(m_1,m_2,k) \\ &\approx \frac{H_{\text{S}}(m_1) + A_{\text{D}}(m_1,k) \text{e}^{-\frac{\text{j}2\pi d_{\text{D}}(k)}{\lambda_{m_1}}}}{H_{\text{S}}(m_2) + A_{\text{D}}(m_2,k) \text{e}^{-\frac{\text{j}2\pi d_{\text{D}}(k)}{\lambda_{m_2}}}} \text{e}^{-\text{j}\Delta\theta(m_1,m_2)}.
\end{aligned}
\end{equation}
Noting that $1/\lambda_{m_2} = (1/\lambda_{m_1})(1 - (\lambda_{m_2}-\lambda_{m_1})/\lambda_{m_2})$, the expression can be rewritten as:
\begin{align}
\label{eq:low_noise_rewritten}
\MoveEqLeft \mathcal{H}(m_1,m_2,k) \nonumber \\
& \approx \frac{(H_{\text{S}}(m_1) + A_{\text{D}}(m_1,k) \text{e}^{-\frac{\text{j}2\pi d_{\text{D}}(k)}{\lambda_{m_1}}})\text{e}^{-\text{j}\Delta\theta(m_1,m_2)}}{H_{\text{S}}(m_2) + A_{\text{D}}(m_2,k) \text{e}^{-\frac{\text{j}2\pi d_{\text{D}}(k)}{\lambda_{m_1}}} \text{e}^{\frac{\text{j}2\pi d_{\text{D}}(k)}{\lambda_{m_1}} \cdot \frac{\lambda_{m_2}-\lambda_{m_1}}{\lambda_{m_2}}}}.
\end{align}
Let the dynamic path length be $d_{\text{D}}(k) = d_0 + \varDelta d(k)$, where $d_0$ is the static component of the path length, and $\varDelta d(k)$ is the small variation from respiration. Given that respiratory displacement (5-12~mm~\cite{wang2016human}) is much smaller than the signal wavelength (around 12.2~cm), we have $\varDelta d(k) \in [-\lambda_{m_1}/10, \lambda_{m_1}/10]$. Moreover, the term $(\lambda_{m_2}-\lambda_{m_1})/\lambda_{m_2} \ll 1$\footnote{For the HT-LTF field, the maximum and minimum subcarrier frequencies are 2470.125~MHz and 2433.875~MHz. The value range of the term $(\lambda_{m_2}-\lambda_{m_1})/\lambda_{m_2}$ is [-0.014894, 0.014675].}. Consequently, the term $\text{e}^{\frac{\text{j}2\pi \varDelta d(k)}{\lambda_{m_1}} \cdot \frac{\lambda_{m_2}-\lambda_{m_1}}{\lambda_{m_2}}}$ approximates to 1\footnote{Under the assumption in the previous footnote, the real part of the exponential term $\text{e}^{\frac{\text{j}2\pi \varDelta d(k)}{\lambda_{m_1}} \cdot \frac{\lambda_{m_2}-\lambda_{m_1}}{\lambda_{m_2}}}$ is in [0.999956, 1], and the imaginary part is in [-0.00936, 0.00936].}. The proposed CSCR can be approximated as:
\begin{align}
\label{eq:low_noise_approximated}
\MoveEqLeft \mathcal{H}(m_1,m_2,k) \nonumber \\
& \approx \frac{(H_{\text{S}}(m_1) + A_{\text{D}}(m_1,k) \text{e}^{-\frac{\text{j}2\pi(d_0+\varDelta d(k))}{\lambda_{m_1}}})\text{e}^{-\text{j}\Delta\theta(m_1,m_2)}}{H_{\text{S}}(m_2) + A_{\text{D}}(m_2,k) \text{e}^{-\frac{\text{j}2\pi(d_0+\varDelta d(k))}{\lambda_{m_1}}} \text{e}^{\frac{\text{j}2\pi d_0}{\lambda_{m_1}}\cdot\frac{\lambda_{m_2}-\lambda_{m_1}}{\lambda_{m_2}}}}.
\end{align}
To reduce the analysis, let us define the following complex variables:
\begin{align}
    \mathcal{A} & = A_{\text{D}}(m_1,k)\text{e}^{-\text{j}\Delta\theta(m_1,m_2)}, \label{eq:def_A}\\
    \mathcal{B} & = H_{\text{S}}(m_1)\text{e}^{-\text{j}\Delta\theta(m_1,m_2)}, \label{eq:def_B}\\
    \mathcal{C} & = A_{\text{D}}(m_2,k)\text{e}^{\frac{\text{j}2\pi d_0}{\lambda_{m_1}}\cdot\frac{\lambda_{m_2}-\lambda_{m_1}}{\lambda_{m_2}}}, \label{eq:def_C}\\
    \mathcal{D} & = H_{\text{S}}(m_2), \label{eq:def_D}\\
    \mathcal{Z} & = \text{e}^{-\frac{\text{j}2\pi(d_0+\varDelta d(k))}.{\lambda_{m_1}}} \label{eq:def_Z}
\end{align}
The expression reduces to a Mobius transformation form:
\begin{equation}
\label{eq:mobius_form}
\mathcal{H}(m_1,m_2,k) \approx \frac{\mathcal{A}\mathcal{Z}+\mathcal{B}}{\mathcal{C}\mathcal{Z}+\mathcal{D}}.
\end{equation}
The variation in $\varDelta d(k)$ due to respiration causes $\mathcal{Z}$ to move on the unit circle in the complex plane.
Leveraging the fundamental property of the Mobius transformation that maps circles to circles, we can rewrite~\eqref{eq:mobius_form} as~\cite{zeng2019farsense}:
\begin{equation}
\label{eq:mobius_rewritten}
\frac{\mathcal{A}\mathcal{Z}+\mathcal{B}}{\mathcal{C}\mathcal{Z}+\mathcal{D}} = \frac{\mathcal{B}\mathcal{C}-\mathcal{A}\mathcal{D}}{\mathcal{C}^2}\cdot\frac{1}{\mathcal{Z}+\frac{\mathcal{D}}{\mathcal{C}}} + \frac{\mathcal{A}}{\mathcal{C}}.
\end{equation}
The proposed CSCR is decomposed into a static component $\mathcal{H}_{\text{S}}(m_1,m_2)$ and a dynamic component $\mathcal{H}_{\text{D}}(m_1,m_2,k)$, as expressed by $\mathcal{H}(m_1,m_2,k) \approx \mathcal{H}_{\text{S}}(m_1,m_2) + \mathcal{H}_{\text{D}}(m_1,m_2,k)$. The static component is the constant offset $\mathcal{H}_{\text{S}}(m_1,m_2) = \frac{\mathcal{A}}{\mathcal{C}}$. The dynamic component is given by $\mathcal{H}_{\text{D}}(m_1,m_2,k) = \frac{\mathcal{B}\mathcal{C}-\mathcal{A}\mathcal{D}}{\mathcal{C}^2}\cdot\frac{1}{\mathcal{Z}+\frac{\mathcal{D}}{\mathcal{C}}}$. 
The term $\frac{1}{\mathcal{Z}+\frac{\mathcal{D}}{\mathcal{C}}}$ (the inversion of the translated circle $\mathcal{Z}$) forms a new circle, which is then scaled and rotated by $\frac{\mathcal{B}\mathcal{C}-\mathcal{A}\mathcal{D}}{\mathcal{C}^2}$. 
Crucially, $\mathcal{H}_{\text{D}}(m_1,m_2,k)$ traces a circular arc in the complex plane, mirroring the behavior of the dynamic component $H_\text{D}(m_1,m_2,k)$ in the ideal CSI model~\eqref{eq:ideal_csi}.

To facilitate the analysis, we define a new Fresnel phase as $\rho_{\text{ratio}}(m_1,m_2,k) = \angle\mathcal{H}_{\text{S}}(m_1,m_2) - \angle\mathcal{H}_{\text{D}}(m_1,m_2,k)$. The relationship between the amplitude $|\mathcal{H}(m_1,m_2,k)|$ and phase $\vartheta(m_1,m_2,k)$ of the proposed CSCR is governed by $\rho_{\text{ratio}}(m_1,m_2,k)$.
The governing equations, as detailed in~\eqref{eq:mag_full} and~\eqref{eq:phase_full}, are structurally analogous to the ideal CSI relationships in \eqref{eq:csi_amplitude} and \eqref{eq:csi_phase_alternative}, confirming that the proposed CSCR restores signal complementarity under low-noise conditions.

\subsection{High-Noise Condition}
In the high-noise condition, $A_{\text{D}}(m,k) \not\gg \dot{\varepsilon}(m,k)$, but we still assume $A_{\text{S}}(m,k) \gg \dot{\varepsilon}(m,k)$. According to~\cite{li2021complexbeat}, the proposed CSCR reduces to:
\begin{align}
\label{eq:high_noise_ratio}
\MoveEqLeft \mathcal{H}(m_1,m_2,k)\nonumber & \\ \approx{} & \frac{H_{\text{S}}(m_1) + A_{\text{D}}(m_1,k)\text{e}^{-\frac{\text{j}2\pi d_{\text{D}}(k)}{\lambda_{m_1}}} + \dot{\varepsilon}(m_1,k)}{H_{\text{S}}(m_2)\text{e}^{\text{j}\Delta\theta(m_1,m_2)}}.
\end{align}
To mitigate the effect of the zero-mean Gaussian noise $\dot{\varepsilon}(m_1,k)$, we apply a moving average over $K_2$ consecutive samples. The resulting averaged CSCR, denoted by $\widehat{\mathcal{H}}(m_1,m_2,k)$, is obtained as:
\begin{equation}
\label{eq:high_noise_averaged}
\begin{split}
\widehat{\mathcal{H}}(&m_1,m_2,k)  \approx \sum_{i=1}^{K_2} \frac{\mathcal{H}(m_1,m_2,(k-1){K_2}+i)}{K_2} \\
& \approx \frac{H_{\text{S}}(m_1) + A_{\text{D}}(m_1,\tilde{k})\text{e}^{-\frac{\text{j}2\pi d_{\text{D}}(\tilde{k})}{\lambda_{m_1}}}}{H_{\text{S}}(m_2)} \text{e}^{-\text{j}\Delta\theta(m_1,m_2)}.
\end{split}
\end{equation}
Following the notation from the low-noise analysis, the averaged CSCR in \eqref{eq:high_noise_averaged} can be approximated as a linear function of $\mathcal{Z}$, expressed with $\mathcal{H}(m_1,m_2,k)$ by $\mathcal{H}(m_1,m_2,k) \approx \frac{\mathcal{A}\mathcal{Z}+\mathcal{B}}{\mathcal{D}}$.
The CSCR is accordingly decomposed into a static component $\mathcal{H}_{\text{S}}(m_1,m_2) = \frac{\mathcal{B}}{\mathcal{D}}$ and a dynamic component $\mathcal{H}_{\text{D}}(m_1,m_2,k) = \frac{\mathcal{A}}{\mathcal{D}}\mathcal{Z}$.
The rotation of $\mathcal{Z}$ on the unit circle causes the dynamic component $\mathcal{H}_{\text{D}}(m_1,m_2,k)$ to trace a corresponding circular path.
The resulting vector composition, a sum of a static vector and a rotating dynamic vector, is also identical to the structure of the ideal CSI model.
Consequently, a new Fresnel phase $\rho_{\text{ratio}}(m_1,m_2,k) = \angle\mathcal{H}_{\text{S}}(m_1,m_2) - \angle\mathcal{H}_{\text{D}}(m_1,m_2,k)$ can be defined, and the relationships in \eqref{eq:mag_full} and \eqref{eq:phase_full} still hold. As a result, the proposed CSCR successfully maintains complementarity even under high-noise conditions.

In summary, the proposed CSCR successfully restores the complementarity of the respiratory signal under both low-noise and high-noise conditions, which completes the proof of \textbf{Proposition}~\ref{thm:blind_spot_free}.
\vphantom{}

\end{CJK}

\bibliographystyle{IEEEtran}
\bibliography{IEEEabrv,myrefs}

\end{document}